\documentclass[a4paper,11pt]{amsart}

\usepackage{amsmath,amsfonts,amssymb,amsthm} 
\usepackage{bbm} 
\usepackage{hyperref,color}

\newtheorem{theorem}{Theorem}[section]
\newtheorem{corollary}[theorem]{Corollary}
\newtheorem{lemma}[theorem]{Lemma}
\newtheorem{property}[theorem]{Property}
\newtheorem{proposition}[theorem]{Proposition}

\theoremstyle{definition}

\theoremstyle{remark}
\newtheorem{assumption}[theorem]{Assumption}
\newtheorem{remark}[theorem]{Remark}
\newtheorem{example}[theorem]{Example}

\theoremstyle{definition}

\numberwithin{equation}{section}

\newcommand{\var}{{\rm var}}

\newcommand{\diag}{{\rm diag}}

\newcommand{\ud}{\,\mathrm{d}}
\newcommand{\Leb}{{\rm Leb}}

\DeclareMathOperator*{\esssup}{ess\,sup}
\DeclareMathOperator*{\essinf}{ess\,inf}

\begin{document}

\title[Malliavin asymptotic expansions]{Malliavin calculus method for
  asymptotic expansion of dual control problems}

\author[Michael Monoyios]{Michael Monoyios \\
Mathematical Institute \\ 
University of Oxford\\
Accepted for publication in SIAM Journal on Financial Mathematics} 

\address{Michael Monoyios \\ 
Mathematical Institute \\ 
University of Oxford \\
Radcliffe Observatory Quarter\\
Woodstock Road\\
Oxford OX2 6GG\\\
UK}

\email{monoyios@maths.ox.ac.uk}

\date{\today}

\thanks{Many thanks to Giuseppe Benedetti for helpful comments, and to
  an Associate Editor and two anonymous referees for careful and
  insightful reading and suggestions which have improved the paper.}

\begin{abstract}

We develop a technique based on Malliavin-Bismut calculus ideas, for
asymptotic expansion of dual control problems arising in connection
with exponential indifference valuation of claims, and with
minimisation of relative entropy, in incomplete markets. The problems
involve optimisation of a functional of Brownian paths on Wiener
space, with the paths perturbed by a drift involving the control. In
addition there is a penalty term in which the control features
quadratically. The drift perturbation is interpreted as a measure
change using the Girsanov theorem, leading to a form of the
integration by parts formula in which a directional derivative on
Wiener space is computed. This allows for asymptotic analysis of the
control problem. Applications to incomplete It\^o process markets are
given, in which indifference prices are approximated in the low risk
aversion limit. We also give an application to identifying the minimal
entropy martingale measure as a perturbation to the minimal martingale
measure in stochastic volatility models.

\end{abstract}

\maketitle

\section{Introduction}
\label{sec:intro}

In this article we use an approach to the Malliavin calculus,
pioneered by Bismut \cite{bismut81}, in which perturbations to Brownian
paths on Wiener space are interpreted as measure changes via the
Girsanov theorem, to derive asymptotic expansions for certain
entropy-weighted stochastic control problems. These problems typically
arise from the dual to investment and indifference pricing problems
under exponential utility.

In the dual approach to investment and hedging problems in incomplete
markets, optimisation problems over trading strategies are converted
to optimisations over probability measures. For example, in
exponential indifference pricing of a European claim with payoff $F$,
the dual control representation of the indifference price is to
maximise the expectation of the payoff subject to an entropic penalty
involving the risk aversion $\alpha$ (as we show in Lemma
\ref{lem:dripnew}). In an It\^o process setting, the optimisation over
measures leads to a problem in which the control is a drift
perturbation to a multi-dimensional Brownian motion. This leads us to
consider control problems of the form (with $\Vert\cdot\Vert$ denoting
the Euclidean norm)
\begin{equation}
\sup_{\varphi}\mathbb
E\left[F\left(W+\varepsilon\int_{0}^{\cdot}\varphi_{s}\ud s\right) -
\frac{1}{2}\int_{0}^{T}\Vert\varphi_{t}\Vert^{2}\ud t\right].
\label{eq:bcp}
\end{equation}
The random variable $F(W+\varepsilon\int_{0}^{\cdot}\varphi_{s}\ud s)$
is a functional of the paths of a drift-perturbed multi-dimensional
Brownian motion $W+\varepsilon\int_{0}^{\cdot}\varphi_{s}\ud s$, where
$\varepsilon$ is a small parameter and $\varphi$ is some adapted
control process. Such a dependence typically arises because $F$
depends on a state variable $X^{(\varepsilon)}$ which is a perturbed
process following
\begin{equation}
\ud X^{(\varepsilon)}_{t} = a_{t}\ud t + b_{t}(\ud W_{t} +
\varepsilon\varphi_{t}\ud t),
\label{eq:dXve}
\end{equation}
with $a,b$ adapted processes.

The idea behind our approach is to view the drift $\varepsilon\varphi$
in (\ref{eq:bcp}) or (\ref{eq:dXve}) as a perturbation to Brownian
paths on Wiener space. For $\varepsilon=0$ the optimal control is
zero, and we suppose that the optimal control for small $\varepsilon$
will be a perturbation around zero. Ideas of Malliavin calculus arise
in differentiating the objective function of the control problem with
respect to $\varepsilon$ at $\varepsilon=0$. This uses Bismut's
\cite{bismut81} approach to the stochastic calculus of variations,
which exploits the Girsanov theorem to translate a drift adjustment
into to a measure change, in order to perform differentiation on path
space. Ultimately, this leads to an asymptotic expansion for the value
function, valid for small $\varepsilon$. In the financial application
to indifference pricing, $\varepsilon^{2}=\alpha$, so one obtains
small risk aversion asymptotics. Similar ideas arise in entropy
minimisation problems, which are the dual to pure investment problems
with exponential utility, and we illustrate an example of this in a
stochastic volatility model, in which the small parameter is
$1-\rho^{2}$, $\rho$ being the correlation between the stock and its
volatility. The power of this approach is that we can obtain results
in non-Markovian models and for quite general path-dependent payoffs.

Entropy-weighted control problems have been treated using variational
principles by Bou\'e and Dupuis \cite{bouedupuis98} (we thank a
referee for pointing out this work to us), with a view to applications
in large deviations theory. The result in \cite{bouedupuis98} is a
representation of the form
\begin{equation}
-\log\mathbb E[{\rm e}^{-g(W)}] = \inf_{v}\mathbb
E\left[\frac{1}{2}\int_{0}^{T}\Vert v_{s}\Vert^{2}\ud s + g\left(W +
\int_{0}^{\cdot}v_{s}\ud s\right)\right].
\label{eq:vcp}
\end{equation}
Bierkens and Kappen \cite{bierkenskappen12} develop the methods in
\cite{bouedupuis98} further and obtain formulae for the optimal
control in (\ref{eq:vcp}) as a Malliavin derivative of the functional
$g(W)$. These papers are in a similar spirit to ours in sharing a
variational point of view. It would be interesting to see if future
work could to relate the results in
\cite{bouedupuis98,bierkenskappen12} to ours.

Utility-based valuation techniques rarely lead to explicit solutions,
and this motivates the interest in approximate solutions. The idea of
using Malliavin calculus methods in asymptotic indifference pricing is
due to Davis \cite{mhad06}. Davis used the approach in a
two-dimensional constant parameter basis risk model, with a traded and
non-traded asset following correlated geometric Brownian motions, and
for a European claim depending only on the final value of the
non-traded asset price. In this model, it turns out that partial
differential equation (PDE) techniques, based on a Cole-Hopf transform
applied to the the Hamilton-Jacobi-Bellman equation of the underlying
utility maximisation problem (see Zariphopoulou \cite{z01}, Henderson
\cite{hend02} and Monoyios \cite{mmdef06}), lead to a closed form
non-linear expectation representation for the indifference price. The
asymptotic expansion obtained by Davis \cite{mhad06} can therefore
also be obtained by applying a Taylor expansion to the non-linear
expectation representation, as carried out in Monoyios
\cite{mmqf04,mmima07}. For this reason, perhaps, the technique
developed by Davis has not been further exploited.

In higher-dimensional models, and in almost all models with random
parameters, the aforementioned Cole-Hopf transform does not
work. Indifference prices and their risk-aversion asymptotics have
been analysed via other methods, notably by backward stochastic
differential equation (BSDE) and bounded-mean-oscillation (BMO)
martingale methods (Mania and Schweizer \cite{manschw05}, Kallsen and
Rheinl\"ander \cite{kr11}) for bounded claims. Monoyios \cite{mmamf10}
derived small risk aversion valuation and hedging results via PDE
techniques, in a random parameter basis risk model generated by
incomplete information on asset drifts. Delbaen {\em et al}
\cite{6auth} and Stricker \cite{str04} used arguments based on a
Fenchel inequality to derive the zero risk aversion limit of the
indifference price. Recently, Henderson and Liang \cite{hl12} have
used BSDE and PDE techniques to derive indifference price
approximations, of a different nature to ours, in a multi-dimensional
non-traded assets model.

The techniques in this paper are different. We resurrect the method
suggested by Davis \cite{mhad06}. The first contribution is to show
that this technique can be significantly generalised, to cover
multi-dimensional It\^o process markets, with no Markov structure
required, and for claims which can be quite general functionals of the
paths of the asset prices. In doing this we elucidate the precise
relation with the Malliavin calculus. The second contribution is to
derive a representation (Proposition \ref{prop:oc}) for the optimal
control in problems of the form (\ref{eq:bcp}), using variational
techniques on Wiener space. This is used in verifying the correct
structure of our asymptotic expansion.

The third contribution is to establish a dual stochastic control
representation (Lemma \ref{lem:dripnew}) of the indifference price
process in a semi-martingale model. This result seems to be the most
compact representation possible. We apply the Malliavin asymptotic
method to this control problem in an It\^o process setting, and derive
the general form of the small risk aversion asymptotic expansion of an
exponential indifference price, recovering the well-known connection
between small risk aversion exponential indifference valuation and
quadratic risk minimisation. Examples are given of multi-asset basis
risk models, and of stochastic correlation in basis risk. Finally, we
show how the technique can be applied to identify the minimal entropy
martingale measure (MEMM) $\mathbb Q^{0}\equiv \mathbb Q_{E}$ as a
perturbation to the minimal martingale measure $\mathbb Q_{M}$ in a
stochastic volatility model, when the stock and volatility are highly
correlated.

Other types of asymptotic expansion for marginal utility-based prices,
in terms of a small holding of claims, have been obtained by Kramkov
and S{\^{\i}}rbu \cite{ks07} and by Kallsen {\em et al}
\cite{kmkv11}. These works use utility functions defined on the
positive half-line, in contrast to the exponential utility function
used in this paper. In stochastic volatility models, Sircar and
Zariphopoulou \cite{sz05} obtain asymptotic expansions for exponential
indifference prices using the fast mean-reversion property of the
volatility process. This approach has been significantly exploited in
many scenarios (see Sircar {\em et al} \cite{fpss11}), and is of a
different nature to our approach.

Malliavin calculus has found application in other areas of
mathematical finance, such as insider trading \cite{imk03}, to
computation of sensitivity parameters \cite{fournie99}, and to other
forms of asymptotic expansion \cite{bgm10}, involving sensitivity with
respect to initial conditions, or with respect to parameters in asset
price dynamics, or to parameters appearing in an expectation, as
opposed to a control.

The rest of the paper is organised as follows. In Section
\ref{sec:ddbfws} we prove a version of the Malliavin
integration-by-parts formula on Wiener space (Lemma \ref{lem:ipcan}),
giving a directional derivative of a Brownian functional. In Section
\ref{sec:cp} this is used to derive our asymptotic expansion (Theorem
\ref{thm:asex}). We use variational methods to characterise the
optimal control (Proposition \ref{prop:oc}) which helps us
characterise the error term in the approximation. The interplay
between directional derivatives on Wiener space, the Malliavin
derivative, and perturbation analysis is exemplified in this
section. In Section \ref{sec:ddrip} we derive, in a locally bounded
semi-martingale model, the dual stochastic control representation of
the indifference price process (Lemma \ref{lem:dripnew}) that forms
the basis of the financial control problems we are interested in. In
Section \ref{sec:iviim} we apply the asymptotic analysis of
indifference valuation in an It\^o process setting. In Section
\ref{sec:examples} we give examples of approximate indifference
valuation in some basis risk models, and we show how the MEMM can be
identified as a perturbation to the minimal martingale measure in a
stochastic volatility model.

\section{Directional derivatives of Brownian functionals on Wiener
  space} 
\label{sec:ddbfws}

In this section we consider perturbations to Brownian paths, and the
ensuing directional derivatives, on Wiener space. This is Bismut's
\cite{bismut81} approach to the Malliavin calculus, and will be used
in asymptotic analysis of control problems in the next section. In
this approach, one deduces a certain invariance principle (see
(\ref{eq:equiv})) by using the Girsanov theorem to translate a drift
perturbation to a Brownian motion into a change of probability
measure. This approach is discussed in Section IV.41 of Rogers and
Williams \cite{rw00}, and Appendix E of Karatzas and Shreve
\cite{ks98}. Nualart \cite{nualart06} is a general treatise on
Malliavin calculus.

The setting uses the canonical basis $(\Omega,{\mathcal F},\mathbb
F:=({\mathcal F}_{t})_{0\leq t\leq T},\mathbb P)$, on which we define
an $m$-dimensional Brownian motion $W$. So,
$\Omega=C_{0}([0,T];\mathbb R^{m})$, the Banach space of continuous
functions $\omega:[0,T]\to\mathbb R^{m}$ null at zero, equipped with
the supremum norm
$\Vert\omega(t)\Vert_{\infty}:=\sup_{t\in[0,T]}\Vert\omega(t)\Vert$,
$\mathbb P$ is Wiener measure, and
$(W_{t}(\omega):=\omega(t))_{t\in[0,T]}$ is $m$-dimensional Brownian
motion with natural filtration $\mathbb F$. The Malliavin calculus is
conventionally introduced with reference to the Hilbert space
$H=L^{2}([0,T],\mathcal{B}([0,T]),\Leb;\mathbb R^{m})$ (we write
$H=L^{2}([0,T];\mathbb R^{m})$ for brevity). An element $h\in H$ is a
function $h:[0,T]\to\mathbb R^{m}$, with norm $\Vert
h\Vert_{H}^{2}=\int_{0}^{T}\Vert h_{t}\Vert^{2}\ud t<\infty$. Then the
Wiener integral $\mathbb W(h)$, defined by
\begin{equation*}
\mathbb W(h) := \sum_{i=1}^{m}\int_{0}^{T}h^{i}_{t}\ud W^{i}_{t} \equiv
\int_{0}^{T}h_{t}\cdot\ud W_{t} \equiv (h\cdot W)_{T},
\end{equation*}
is an isonormal Gaussian process. That is, the linear isometry $\mathbb
W:H\to L^{2}[(\Omega,\mathcal{F},\mathbb P);\mathbb R]$ is such that
$\mathbb W=(\mathbb W(h))_{h\in H}$ is a centred family of
Gaussian random variables with $\mathbb E[\mathbb W(h)]=0$ and
\begin{equation*}
\mathbb E[\mathbb W(h)\mathbb W(g)] = \langle h,g\rangle_{H} =
\int_{0}^{T}h_{t}\cdot g_{t}\ud t =
\sum_{i=1}^{m}\int_{0}^{T}h^{i}_{t}g^{i}_{t}\ud t.   
\end{equation*}
For $\varphi\in H=L^{2}([0,T];\mathbb R^{m})$, the Cameron-Martin
subspace $\mathcal{CM}\subset\Omega=C_{0}([0,T];\mathbb R^{m})$ is
composed of absolutely continuous functions $\Phi:[0,T]\to\mathbb
R^{m}$ with square-integrable derivative $\varphi$. That is,
\begin{equation*}
\Phi_{t} := \int_{0}^{t}\varphi_{s}\ud s, \quad
\int_{0}^{t}\Vert\varphi_{s}\Vert^{2}\ud s < \infty, \quad 0\leq t\leq T.
\end{equation*} 
One transports the Hilbert space structure of $H$ to $\mathcal{CM}$ by
defining
\begin{equation*}
\langle\Phi,\Psi\rangle_{\mathcal{CM}} := \langle\varphi,\psi\rangle_{H} =
\int_{0}^{T}\varphi_{t}\cdot\psi_{t}\ud t, \quad \Psi :=
\int_{0}^{\cdot}\psi_{s}\ud s, 
\end{equation*}
so $\mathcal{CM}$ is isomorphic to $H$. If $F$ is
Malliavin-differentiable, then there exists an $H$-valued random
variable, so a measurable (but not necessarily adapted) process
$(D_{t}F)_{t\in[0,T]}$, such that for $\Phi\in\mathcal{CM}$ we have
the integration-by-parts formula
\begin{equation*}
\mathbb E[\langle DF,h\rangle_{H}] = \mathbb E[F\mathbb W(h)],  
\end{equation*}
or
\begin{equation}
\mathbb E\left[\int_{0}^{T}D_{t}F\cdot\varphi_{t}\ud t\right] =
\mathbb E\left[F\int_{0}^{T}\varphi_{t}\cdot\ud W_{t}\right],
\label{eq:ibpcm}
\end{equation}
and $\langle DF,h\rangle_{H}$ has properties of a directional
derivative. This will be transparent in the Bismut approach to the
Malliavin calculus, which we outline below.

\subsection{The Bismut approach}
\label{subsec:bismut}

Bismut \cite{bismut81} developed an alternative version of the
stochastic calculus of variations, in which the left-hand-side of
(\ref{eq:ibpcm}) is a directional derivative on Wiener space, and
which allows for $\varphi$ to be a previsible process.

We have a square-integrable functional $F(W)$ of the Brownian paths
$W$, that is, an $\mathcal{F}_{T}$-measurable map $F:\Omega\to\mathbb
R$ satisfying 
\begin{equation}
\mathbb E[F^{2}(W)] < \infty.
\label{eq:sqiF}
\end{equation}
Let $\Phi\in C^{1}_{0}([0,T];\mathbb R^{m})\subset \Omega$, with
$\Phi:=\int_{0}^{\cdot}\varphi_{s}\ud s$ for some previsible process
$\varphi$. We are interested in defining a directional derivative of
$F$ in the direction $\Phi$.

The first variation (or G\^ateaux variation) $\delta F(W;\Phi)$
of $F$ at $W\in\Omega$ in the direction $\Phi$ is the limit, if it
exists, given by
\begin{equation*}
\delta F(W;\Phi) := \lim_{\varepsilon\to
0}\frac{1}{\varepsilon}[F(W + \varepsilon\Phi) - F(W)] =
\frac{\ud}{\ud\varepsilon}\left.[F(W +
\varepsilon\Phi)]\right\vert_{\varepsilon=0}.
\end{equation*}
(See Luenberger \cite{luenberger69} (Chapter 7) or Wouk \cite{wouk79}
(Chapter 12) for more on this and other notions of differentiation of
non-linear maps in Banach spaces.) The first variation is homogeneous
in the perturbation $\Phi$: $\delta F(W;c\Phi)=c\delta F(W;\Phi)$ for
$c\in\mathbb R$. We are interested in the case when $F$ is such that
the first variation is also linear in $\Phi$. To this end, we impose
the following conditions on $F$, similar to those used in Appendix E
of Karatzas and Shreve \cite{ks98}.

\begin{assumption}
\label{ass:condF}

\begin{itemize}

\item[(i)] $F$ satisfies square-integrability condition
  (\ref{eq:sqiF}).

\item[(ii)] There exists a non-negative Brownian functional $k$
satisfying $\mathbb E[k^{2}(W)]<\infty$ and a function
$g:[0,\infty)\to[0,\infty)$ satisfying $\limsup_{\varepsilon\downarrow 
0}(g(\varepsilon)/\varepsilon)<\infty$, such that for 
$W,\Phi\in\Omega$,
\begin{equation}
|F(W+\Phi) - F(W)| \leq k(W)g(\Vert\Phi\Vert_{\infty}). 
\label{eq:condF}
\end{equation}

\item[(iii)] There exists a kernel $\partial
  F(\omega;\cdot)\equiv\partial F(W;\cdot):\Omega\to\mathbb M$, where
  $\mathbb M$ is the set of $m$-dimensional finite Borel measures on
  $[0,T]$, such that for each $\Phi\in C^{1}_{0}([0,T];\mathbb
  R^{m})\subset\Omega$ we have
\begin{equation}
\lim_{\varepsilon\to 0}\frac{1}{\varepsilon}\left[F(W+\varepsilon\Phi)
- F(W)\right] = \int_{0}^{T}\Phi_{t}\cdot\partial F(W;\ud t), \quad
\mbox{for almost all $W\in\Omega$}.
\label{eq:partialF}
\end{equation}

\end{itemize}
  
\end{assumption}

Note, in particular, that (\ref{eq:partialF}) implies
\begin{equation}
F(W+\varepsilon\Phi) = F(W) +
\varepsilon\int_{0}^{T}\Phi_{t}\cdot\partial F(W;\ud t) +
o(|\varepsilon|\Vert\Phi\Vert_{\infty}).   
\label{eq:asymp1}
\end{equation}
Using $\Phi=\int_{0}^{\cdot}\varphi_{s}\ud s$ on the right-hand-side
of (\ref{eq:partialF}), we may integrate by parts to obtain the
equivalent form
\begin{equation}
\lim_{\varepsilon\to
0}\frac{1}{\varepsilon}\left[F\left(W +
\varepsilon\int_{0}^{\cdot}\varphi_{s}\ud s\right) 
- F(W)\right] = \int_{0}^{T}\partial F(W;(t,T])\cdot\varphi_{t}\ud t.
\label{eq:partialF2}
\end{equation}
In particular, we then have the analogue of (\ref{eq:asymp1}):
\begin{equation}
F\left(W+\varepsilon\int_{0}^{\cdot}\varphi_{s}\ud s\right) = F(W) +
\varepsilon\int_{0}^{T}\partial F(W;(t,T])\cdot\varphi_{t}\ud t +
o(|\varepsilon|\Vert\varphi\Vert_{\infty}).
\label{eq:asymp2}
\end{equation}
Rogers and Williams \cite{rw00} (Section IV.41) make the observations
that the condition (\ref{eq:partialF}) in Assumption \ref{ass:condF}
is automatically satisfied if $F$ is Fr\'echet differentiable with
bounded derivative, and in that case $\partial F\equiv F^{\prime}$,
where the Fr\'echet derivative $F^{\prime}(W;\cdot)$ is a bounded
linear functional on $\Omega$ (that is, a measure with finite total
variation, an element of the dual space $\Omega^{\prime}$). But there
are functionals where differentiability fails but (\ref{eq:partialF})
holds (\cite{rw00}, Section IV.41 has such examples).

Our particular interest will be in the functional $\mathbb
E[F(W+\varepsilon\Phi)]$ and the associated variation
\begin{equation*}
\lim_{\varepsilon\to 0}\frac{1}{\varepsilon}\mathbb E[F(W +
\varepsilon\Phi) - F(W)] = \frac{\ud}{\ud\varepsilon}\left.\mathbb
E[F(W + \varepsilon\Phi)]\right\vert_{\varepsilon=0}.  
\end{equation*}
It turns out that one can make sense of this limit, resulting in a
version of the integration-by-parts formula (\ref{eq:ibpcm}) which
holds regardless of whether $F$ is Malliavin differentiable. This is
given in Lemma \ref{lem:ipcan} further below. 


\subsubsection{The Clark formula}
\label{subsubsec:clarkf}

The classical result of Clark \cite{clark70} relates the kernel
$\partial F$ to the progressively measurable integrand $\psi$
(satisfying $\mathbb E[\int_{0}^{T}\Vert\psi_{t}\Vert^{2}\ud
t]<\infty$) in the martingale representation of $F(W)$:
\begin{equation}
F(W) = \mathbb E[F(W)] + \int_{0}^{T}\psi_{t}\cdot\ud W_{t}.
\label{eq:mrF}
\end{equation}
The Clark formula gives $\psi$ as the predictable projection of the
measurable (but not necessarily adapted) process $(\partial
F(W;(t,T]))_{0\leq t\leq T}$. In other words,
\begin{equation}
\psi_{t} = \mathbb E[\partial F(W;(t,T])|\mathcal{F}_{t}], \quad 0\leq
t\leq T.
\label{eq:clark}  
\end{equation}
This result is proven in Appendix E of Karatzas and Shreve \cite{ks98}
and in Section IV.41 of Rogers and Williams \cite{rw00}, using similar
methods to those that we shall employ in the proof of Lemma
\ref{lem:ipcan} below.

\begin{lemma}[Directional derivative on Wiener space]
\label{lem:ipcan}

Let $F\equiv F(W)$ be a functional of the Brownian paths $W$ on the
Banach space $\Omega=C_{0}([0,T];\mathbb R^{m})$ satisfying Assumption
\ref{ass:condF}. Let $\varphi$ be a bounded previsible process, with
$\Phi\in C^{1}_{0}([0,T];\mathbb R^{m})\subset\Omega$ defined by
$\Phi:=\int_{0}^{\cdot}\varphi_{s}\ud s$. Then the map
$\varepsilon\to\mathbb E[F(W+\varepsilon\Phi)]$ is differentiable,
with derivative
\begin{equation}
\frac{\ud}{\ud\varepsilon}\left.\mathbb
E[F(W+\varepsilon\Phi)]\right\vert_{\varepsilon=0} 
= \mathbb E\left[F(W)(\varphi\cdot W)_{T}\right].
\label{eq:ddws}
\end{equation}
Moreover, if $\varphi=c\widetilde{\varphi}$ for some fixed
$\widetilde{\varphi}$ and $c\in\mathbb R$, then
\begin{equation}
\mathbb E[F(W+\varepsilon\Phi) - F(W) - \varepsilon F(W)(\varphi\cdot
W)_{T}] \sim O(c^{2}\varepsilon^{2}). 
\label{eq:errortermc}
\end{equation}

\end{lemma}

A form of Lemma \ref{lem:ipcan} appears in Davis \cite{mhad06} (his
Lemma 3) in a one-dimensional set-up, with a functional dependent only
on the final value of a diffusion. Fourni\'e {\em et al}
\cite{fournie99} have results of a similar nature in the context of
perturbations arising from variations in the drift or diffusion
coefficients of Markov SDEs (see, for instance, Proposition 3.1 in
\cite{fournie99}).

To prove Lemma \ref{lem:ipcan} we will need the following property of
exponential martingales. 

\begin{lemma}
\label{lem:l2convM}

For a bounded previsible process $\varphi$ and $\varepsilon\in\mathbb
R$, define the exponential martingale
\begin{equation}
M^{(\varepsilon)}_{t} := \mathcal{E}(-\varepsilon\varphi\cdot W)_{t}
:= \exp\left(-\varepsilon\int_{0}^{t}\varphi_{s}\cdot \ud W_{s} - 
\frac{1}{2}\varepsilon^{2}\int_{0}^{t}\Vert\varphi_{s}\Vert^{2}\ud
s\right), \quad 0\leq t\leq T.
\label{eq:Medef}
\end{equation}
Then we have
\begin{equation}
\label{eq:limit1}
\lim_{\varepsilon\to 0}\mathbb
E\left[\int_{0}^{t}(1 - M^{(\varepsilon)}_{s})^{2}\ud s\right] = 0,
\quad 0\leq t\leq T,
\end{equation}
and
\begin{equation}
\frac{1}{\varepsilon}(1 - M^{(\varepsilon)}_{t}) \to
(\varphi\cdot W)_{t}, \quad \mbox{in $L^{2}$, as
$\varepsilon\to 0$, for every $t\in[0,T]$}.
\label{eq:limit}
\end{equation}

\end{lemma}

\begin{proof}

Since $\varphi$ is bounded, Novikov's criterion is satisfied and
$M^{(\varepsilon)}$ is a martingale. Using the representation
\begin{equation}
M^{(\varepsilon)}_{t} = 1 -
\varepsilon\int_{0}^{t}M^{(\varepsilon)}_{s}\varphi_{s}\cdot\ud W_{s},
\quad 0\leq t\leq T, 
\label{eq:Msoln}
\end{equation}
the stochastic integral is a martingale and we have
\begin{equation}
\mathbb
E\left[\int_{0}^{t}(M^{(\varepsilon)}_{s})^{2}\Vert\varphi_{s}\Vert^{2}\ud
  s\right] < \infty, \quad 0\leq t\leq T.  
\label{eq:deic2}
\end{equation}
Using (\ref{eq:Msoln}) along with the It\^o isometry, we have, for any
$t\in[0,T]$, 
\begin{equation*}
\mathbb E\left[\int_{0}^{t}(1 - M^{(\varepsilon)}_{s})^{2}\ud s\right]
= \varepsilon^{2}\mathbb 
E\left[\int_{0}^{t}\int_{0}^{s}(M^{(\varepsilon)}_{u})^{2}\Vert\varphi_{u}\Vert^{2}\ud
u\ud s\right].  
\end{equation*}
By (\ref{eq:deic2}), the expectation on the right-hand-side is finite
for any value of $\varepsilon$. Hence, letting $\varepsilon\to 0$ we
obtain (\ref{eq:limit1}).

Using (\ref{eq:Msoln}) and the It\^o isometry once again, we
compute, for any $t\in[0,T]$,
\begin{equation*}
\mathbb E\left[
\left(\frac{1}{\varepsilon}\left(1 - M^{(\varepsilon)}_{t}\right) -
(\varphi\cdot W)_{t}\right)^{2}\right] 
= \mathbb E\left[\int_{0}^{t}(1 -
M^{(\varepsilon)}_{s})^{2}\Vert\varphi_{s}\Vert^{2}\ud s\right], 
\end{equation*}
which, using (\ref{eq:limit1}) and the fact that $\varphi$ is bounded,
converges to zero as $\varepsilon\to 0$, and this gives
(\ref{eq:limit}).

\end{proof}

\begin{proof}[Proof of Lemma \ref{lem:ipcan}]

We use a version of arguments found in some proofs of the Clark
representation formula (see, for instance, Appendix E of Karatzas and
Shreve \cite{ks98} or the proof of Theorem IV.41.9 in Rogers and
Williams \cite{rw00}).

For $\varepsilon\in\mathbb R$ and $\varphi$ previsible and bounded,
define the probability measure $\mathbb P^{(\varepsilon)}$ by
\begin{equation*}
\frac{\ud\mathbb P^{(\varepsilon)}}{\ud\mathbb P} = M^{(\varepsilon)}_{T},
\end{equation*}
where $M^{(\varepsilon)}$ is the exponential martingale defined in
(\ref{eq:Medef}). By the Girsanov Theorem, $W+\varepsilon\Phi$ is
Brownian motion under $\mathbb P^{(\varepsilon)}$, so that with
$\mathbb E^{(\varepsilon)}$ denoting expectation under $\mathbb
P^{(\varepsilon)}$,
\begin{equation}
\mathbb E[F(W)] = \mathbb E^{(\varepsilon)}\left[F(W +
\varepsilon\Phi)\right] = \mathbb E[M^{(\varepsilon)}_{T}F(W + \varepsilon\Phi)].
\label{eq:equiv}
\end{equation}
This invariance principle underlies Bismut's approach to the Malliavin
calculus.

Re-write (\ref{eq:equiv}) as
\begin{equation}
\mathbb E\left[\frac{F(W+\varepsilon\Phi) -
F(W)}{\varepsilon}\right] = \mathbb
E\left[\frac{1 - M^{(\varepsilon)}_{T}}{\varepsilon}F(W)\right] 
+ \mathbb E\left[\frac{F(W+\varepsilon\Phi) - 
F(W)}{\varepsilon}(1 - M^{(\varepsilon)}_{T})\right]. 
\label{eq:equiv2}
\end{equation}
We differentiate $\mathbb E[F(W+\varepsilon\Phi)]$ with respect to
$\varepsilon$ at $\varepsilon=0$ by considering what happens when we
let $\varepsilon\to 0$ in (\ref{eq:equiv2}). To this end, subtract
$\mathbb E[F(W)(\varphi\cdot W)_{T}]$ from both sides, to compute
\begin{eqnarray}
&& \mathbb E\left[\frac{1}{\varepsilon}(F(W+\varepsilon\Phi) -
F(W)) - F(W)(\varphi\cdot W)_{T}\right] \nonumber \\
& = & \mathbb
E\left[\left(\frac{1 - M^{(\varepsilon)}_{T}}{\varepsilon} -
(\varphi\cdot W)_{T}\right)F(W)\right]
+ \mathbb E\left[\frac{F(W+\varepsilon\Phi) - F(W)}{\varepsilon}(1 -
M^{(\varepsilon)}_{T})\right].  
\label{eq:equiv3}
\end{eqnarray}
Now take the limit $\varepsilon\to 0$ in
(\ref{eq:equiv3}). Using conditions (i) and (ii) in Assumption
(\ref{ass:condF}), the dominated convergence theorem and the
Cauchy-Schwarz inequality, the last term on the right-hand-side is
bounded by
\begin{equation*}
\mathbb E[k(W)(g(|\varepsilon|\Vert\Phi\Vert_{\infty})/|\varepsilon|)|1 -
M^{(\varepsilon)}_{T}|] \leq 
K(\mathbb E[(1 - M^{(\varepsilon)}_{T})^{2}])^{1/2}, \quad \mbox{for
some constant $K$}, 
\end{equation*}
which converges to zero as $\varepsilon\to 0$, because of
(\ref{eq:limit1}).

Next consider the first term on the right-hand-side of
(\ref{eq:equiv3}). Using the square-integrability of $F$ and the
Cauchy-Schwarz inequality, we have
\begin{equation*}
\left(\mathbb E\left[\left(\frac{1 -
M^{(\varepsilon)}_{T}}{\varepsilon} - (\varphi\cdot
W)_{T}\right)F(W)\right]\right)^{2} \leq K\mathbb
E\left[\left(\frac{1 - M^{(\varepsilon)}_{T}}{\varepsilon} - 
(\varphi\cdot W)_{T}\right)^{2}\right], 
\end{equation*}
for some constant $K$. This converges to zero as $\varepsilon\to 0$,
on using (\ref{eq:limit}). Thus the right-hand-side, and hence the
left-hand-side, of (\ref{eq:equiv3}) converges to zero as
$\varepsilon\to 0$ and this establishes (\ref{eq:ddws}), the first
part of the lemma.

To establish (\ref{eq:errortermc}), we apply the same arguments to
(\ref{eq:equiv3}) multiplied by $\varepsilon$. We have
\begin{eqnarray*}
&& \mathbb E\left[F(W+\varepsilon\Phi) -
F(W) - \varepsilon F(W)(\varphi\cdot W)_{T}\right] \nonumber \\
& = & \mathbb
E\left[\left(1 - M^{(\varepsilon)}_{T} - \varepsilon(\varphi\cdot
W)_{T}\right)F(W)\right]
+ \mathbb E\left[(F(W+\varepsilon\Phi) - 
F(W))(1 - M^{(\varepsilon)}_{T})\right]. 
\end{eqnarray*}
We examine how each of the terms on the right-hand-side scale for small
$\varepsilon$ and $\varphi=c\widetilde{\varphi}$. Using the
representation (\ref{eq:Msoln}), the second term satisfies
\begin{equation*}
\mathbb E\left[(F(W+\varepsilon\Phi) - 
F(W))(1 - M^{(\varepsilon)}_{T})\right] \leq \varepsilon\mathbb
E\left[k(W)g(|\varepsilon|\Vert\Phi\Vert_{\infty})
\left\vert\int_{0}^{T}M^{(\varepsilon)}_{t}\varphi_{t}\cdot\ud
W_{t}\right\vert\right],  
\end{equation*}
and so for $\varphi=c\widetilde{\varphi}$ this term is of
$O(c^{2}\varepsilon^{2})$, on invoking the properties of $g(\cdot)$ in
Assumption \ref{ass:condF} (ii). For the first term, using the
representation (\ref{eq:Msoln}) for $\varphi=c\widetilde{\varphi}$, we
have
\begin{equation*}
\mathbb E\left[\left(1 - M^{(\varepsilon)}_{T} - \varepsilon(\varphi\cdot
W)_{T}\right)F(W)\right] = c^{2}\varepsilon^{2}\mathbb
E\left[F(W)\int_{0}^{T}\frac{1}{c\varepsilon}(M^{(\varepsilon)}_{t}
- 1)\widetilde{\varphi}_{t}\cdot\ud W_{t}\right],
\end{equation*}
which is is of $O(c^{2}\varepsilon^{2})$, on using
(\ref{eq:limit}). Hence (\ref{eq:errortermc}) is established.

\end{proof}

\begin{remark}
\label{rem:boundedphi}

The boundedness condition on $\varphi$ in Lemma \ref{lem:ipcan} can be
relaxed. A Novikov condition on $\varepsilon\varphi$ would suffice, so
that the stochastic exponential $M^{(\varepsilon)}$ in
(\ref{eq:Medef}) is a martingale. This remark also pertains to Lemma
\ref{lem:l2convM}.

\end{remark}

\subsection{Relation with the Malliavin derivative}
\label{subsubsec:rwmd}

We can connect the Malliavin derivative of $F$ (when this exists) to
the kernel $\partial F$ in the condition (\ref{eq:partialF}) and to
the directional derivative in Lemma \ref{lem:ipcan}. 

First, note that the right-hand-side of (\ref{eq:ddws}) has the same
structure as the right-hand-side of (\ref{eq:ibpcm}), so Lemma
\ref{lem:ipcan} is a version of the integration-by-parts formula,
generalised to Bismut's set-up, when
$\Phi=\int_{0}^{\cdot}\varphi_{s}\ud s$ need not be restricted to
elements of $\mathcal{CM}$.

Now, with $\Phi=\int_{0}^{\cdot}\varphi_{s}\ud s$, take the limit
$\varepsilon\to 0$ in (\ref{eq:equiv2}), using the conditions in
Assumption \ref{ass:condF} and the Dominated Convergence Theorem (a
similar procedure is used in \cite{ks98,rw00} in proving the
Clark representation formula), to obtain
\begin{equation}
\mathbb E\left[\int_{0}^{T}\Phi_{t}\cdot\partial F(W;\ud t)\right] =
\mathbb E\left[F(W)\int_{0}^{T}\varphi_{t}\cdot\ud W_{t}\right],
\label{eq:ibpartial0}
\end{equation}
Using $\Phi=\int_{0}^{\cdot}\varphi_{s}\ud s$ and integrating by parts
on the left-hand-side as was done to obtain (\ref{eq:partialF2}), we
convert (\ref{eq:ibpartial0}) to the equivalent form
\begin{equation}
\mathbb E\left[\int_{0}^{T}\partial F(W;(t,T])\cdot\varphi_{t}\ud t\right] =
\mathbb E\left[F(W)\int_{0}^{T}\varphi_{t}\cdot\ud W_{t}\right].
\label{eq:ibpartial}
\end{equation}
Comparing with (\ref{eq:ddws}), we see that the left-hand-side of
(\ref{eq:ibpartial}) is just another way to write the directional
derivative in Lemma \ref{lem:ipcan}. Note that if we use the
martingale representation (\ref{eq:mrF}) of $F$ on the right-hand-side
of (\ref{eq:ibpartial}) we obtain the Clark formula (\ref{eq:clark}).

In the case that $\Phi\equiv\int_{0}^{\cdot}\varphi_{s}\ud s$ is an
element of the Cameron-Martin space $\mathcal{CM}$, and for
Malliavin-differentiable $F$, the right-hand-side of (\ref{eq:ddws})
or (\ref{eq:ibpartial}) is also the right-hand-side of the
integration-by-parts formula (\ref{eq:ibpcm}), so in this case the
kernel $\partial F$ is related to the Malliavin derivative according
to
\begin{equation*}
\partial F(W;(t,T]) = D_{t}F(W), \quad 0\leq t\leq T,
\end{equation*}
and (\ref{eq:ibpartial}) is the integration-by-parts formula.  So when
$F$ is Malliavin-differentiable and
$\Phi\in\mathcal{CM}\subset\Omega$, the directional derivative in
(\ref{eq:ddws}) is also given by $\mathbb
E\left[\int_{0}^{T}D_{t}F\cdot\varphi_{t}\ud t\right]$. But Lemma
\ref{lem:ipcan} is valid when $F$ is not necessarily
Malliavin-differentiable and for previsible $\varphi$, with
$\Phi\equiv\int_{0}^{\cdot}\varphi_{s}\ud s$ not necessarily in
$\mathcal{CM}$.

\section{Malliavin asymptotics of a control problem}
\label{sec:cp}

In this section we describe a control problem and analyse it via
variational principles and Bismut-Malliavin asymptotics. How this type
of problem arises in a financial model will be described in subsequent
sections.

We have a canonical basis $(\Omega,\mathcal{F},\mathbb
F=(\mathcal{F}_{t})_{0\leq t\leq T},\mathbb P)$, on which is defined
an $m$-dimensional Brownian motion $W$. A square-integrable random
variable $F$ is a functional of the paths of the perturbed Brownian
motion $W+\varepsilon\int_{0}^{\cdot}\varphi_{s}\ud s$, where
$\varepsilon\in\mathbb R$ is a small parameter and $\varphi$ is a
control process satisfying $\int_{0}^{T}\Vert\varphi_{t}\Vert^{2}\ud
t<\infty$, and such that $\mathcal{E}(|\varepsilon|\varphi\cdot W)$ is
a martingale. (A Novikov condition $\mathbb
E[\exp\left(\frac{1}{2}\varepsilon^{2}\int_{0}^{T}\Vert\varphi_{t}\Vert^{2}\ud
t\right)]<\infty$ would suffice). Denote by $\mathcal{A}$ the set of
such controls.

The control problem we are interested in is to maximise an objective
functional $G(\varphi)$, defined by
\begin{equation}
G(\varphi) := \mathbb E\left[F\left(W +
\varepsilon\int_{0}^{\cdot}\varphi_{s}\ud s\right) - 
\frac{1}{2}\int_{0}^{T}\Vert\varphi_{t}\Vert^{2}\ud t\right].
\label{eq:Gfunc}
\end{equation}
The value function is
\begin{equation}
p:= G(\varphi^{*}) = \sup_{\varphi\in\mathcal{A}}G(\varphi),  
\label{eq:objf}
\end{equation}
for some optimal control $\varphi^{*}$.  

As we shall see in Section \ref{sec:iviim}, in finance this type of
control problem typically arises because $F$ is a functional of the
Brownian paths through dependence on some perturbed state variable
$X\equiv X^{(\varepsilon)}\in \mathbb R^{m}$, following an It\^o
process
\begin{equation} 
\ud X^{(\varepsilon)}_{t} = a_{t}\ud t + b_{t}(\ud W_{t} +
\varepsilon\varphi_{t}\ud t), 
\label{eq:xeps2}
\end{equation} 
for some $m$-dimensional adapted process $a$ satisfying
$\int_{0}^{T}\Vert a_{t}\Vert\ud t<\infty$ and $m\times m$ adapted
matrix process $b$ satisfying $\int_{0}^{T}\Vert a_{t}\Vert\ud
t<\infty$ and $\int_{0}^{T}b_{t}b^{\top}_{t}\ud t<\infty$. In this
section we shall not require a state process $X^{(\varepsilon)}$.

The idea behind the asymptotic expansion is to treat
$\varepsilon\varphi$ as a perturbation to the Brownian paths.
We suppose that, for small $\varepsilon$, the optimal control
$\varphi^{*}$ will be small. We expand the objective functional in
(\ref{eq:Gfunc}) about $\varepsilon=0$ using Lemma
\ref{lem:ipcan}. Naturally, for $\varepsilon=0$ the functional
$F(W+\varepsilon\int_{0}^{\cdot}\varphi_{s}\ud s)$ loses all
dependence on the control $\varphi$, so in this case optimal control
is zero, and the leading order term will be $\mathbb E[F(W)]$.
Ultimately, this leads to the main result below, a perturbative
solution to the control problem (\ref{eq:objf}).

\begin{theorem}
\label{thm:asex}

On the canonical basis
$(\Omega,\mathcal{F},(\mathcal{F}_{t})_{t\in[0,T]},\mathbb P)$, define
an $m$-dimensional Brownian motion $W$. Let
$\Phi:=\int_{0}^{\cdot}\varphi_{s}\ud s\in\Omega$ be such that
$\int_{0}^{T}\Vert\varphi_{t}\Vert^{2}\ud t<\infty$ and
$\mathcal{E}(|\varepsilon|\varphi\cdot W)$ is a martingale, with
$\varepsilon\in\mathbb R$ a small parameter. Denote the set of such
$\varphi$ by $\mathcal{A}$. Let $F(W +
\varepsilon\int_{0}^{\cdot}\varphi_{s}\ud s)$ be a functional of the
paths of the perturbed Brownian motion
$W+\varepsilon\int_{0}^{\cdot}\varphi_{s}\ud s$, and satisfying
Assumption \ref{ass:condF}. Then the control problem with value
function (\ref{eq:objf}) has asymptotic value given by
\begin{equation}
p = \mathbb E[F(W)] + \frac{1}{2}\varepsilon^{2}\mathbb
E\left[\int_{0}^{T}\Vert\psi_{t}\Vert^{2}\ud
t\right] + O(\varepsilon^{4}), 
\label{eq:asex}
\end{equation}
where $\psi$ is the integrand in the martingale representation
(\ref{eq:mrF}) of $F(W)$.

\end{theorem}

\begin{remark}[Heuristics]
\label{rem:heuristics}

Before proving the theorem, we outline the ideas underlying the proof
in a simple setting where ordinary calculus can replace variational
calculus on Wiener space.

Consider maximising, over a scalar variable $\varphi$, a smooth
function $J(\varepsilon,\varphi)$ given by
\begin{equation}
J(\varepsilon,\varphi) := f(x+\varepsilon\varphi) -
\frac{1}{2}\varphi^{2},
\label{eq:Jfunc}
\end{equation}
for some smooth function $f$, and with $\varepsilon$ a small
parameter. The optimiser of this problem satisfies 
\begin{equation}
\varphi^{*} = \varepsilon f^{\prime}(x+\varepsilon\varphi^{*}),
\label{eq:vphistar}
\end{equation}
so is of course zero for $\varepsilon=0$. If we seek a power series
approximation for $\varphi^{*}$, writing
$\varphi^{*}=\sum_{k=1}^{\infty}\varepsilon^{k}\varphi^{(k)}$ for some
coefficients $\varphi^{(k)}$, then using this in (\ref{eq:vphistar})
along with a Taylor expansion of $f^{\prime}(x+\varepsilon\varphi)$
gives
\begin{equation*}
\varphi^{*} = \varepsilon f^{\prime}(x)(1 +
\varepsilon^{2}f^{\prime\prime}(x)) + O(\varepsilon^{5}).
\end{equation*}
In particular, the first two terms in $\varphi^{*}$ are linear and
cubic in $\varepsilon$. With the given structure of the objective
function in (\ref{eq:Jfunc}), this implies that the maximum has
asymptotic expansion given by
\begin{equation*}
J(\varepsilon,\varphi^{*}) = f(x) +
\frac{1}{2}\varepsilon^{2}(f^{\prime}(x))^{2} + O(\varepsilon^{4}).
\end{equation*}
But this is the same value as is obtained by maximising the
linear-in-$\varepsilon$ approximation to $J(\varepsilon,\varphi)$:
\begin{equation*}
J(\varepsilon,\varphi) = f(x) + \varepsilon\varphi f^{\prime}(x) -
\frac{1}{2}\varphi^{2} + O(\varepsilon^{2}\varphi^{2}),
\end{equation*}
which is maximised by $\widehat{\varphi}=\varepsilon f^{\prime}(x)$,
yielding
\begin{equation*}
J(\varepsilon,\widehat{\varphi}) = f(x) +
\frac{1}{2}\varepsilon^{2}(f^{\prime}(x))^{2} + O(\varepsilon^{4}),
\end{equation*}
so that $J(\varepsilon,\varphi^{*})=J(\varepsilon,\widehat{\varphi})$
to order $\varepsilon^{2}$, with the remainder being of order
$\varepsilon^{4}$ in both cases.

We shall see that a similar structure underlies the proof of Theorem
\ref{thm:asex} which we give further below, after some preparation.
 
\end{remark}

The following result is the analogue of (\ref{eq:vphistar}) for
differentiation (in a variational sense) of the the objective
functional of the control problem with respect to the control
$\varphi$. We will use this later in establishing the asymptotic
expansion of Theorem \ref{thm:asex}.

\begin{proposition}
\label{prop:oc}

Assume the same setting as in Theorem \ref{thm:asex}. The optimal
control $\varphi^{*}$ for the problem with value function
(\ref{eq:objf}) satisfies 
\begin{equation}
\varphi^{*}_{t} = \varepsilon\mathbb E\left[\left.\partial F\left(W +
\varepsilon\int_{0}^{\cdot}\varphi^{*}_{s}\ud
s;(t,T]\right)\right\vert\mathcal{F}_{t}\right], \quad 0\leq t\leq T,
\label{eq:phistar}
\end{equation}
where $\partial F(W + \varepsilon\int_{0}^{\cdot}\varphi^{*}_{s}\ud
s;\cdot)$ is the kernel in (\ref{eq:partialF}), evaluated at $W +
\varepsilon\int_{0}^{\cdot}\varphi^{*}_{s}\ud s\in\Omega$.
  
\end{proposition}

\begin{proof}

Recall the conditions (\ref{eq:condF}) and (\ref{eq:partialF}) in
Assumption \ref{ass:condF}. We shall use these to differentiate, in a
variational manner akin to our development of Lemma \ref{lem:ipcan},
the objective functional (\ref{eq:Gfunc}) of the control problem with
respect to the control $\varphi$. 

Consider varying $\varphi$ in (\ref{eq:Gfunc}). To this end, for
$\gamma\in\mathbb R$ a small parameter and
$\Xi=\int_{0}^{\cdot}\xi_{s}\ud s\in\Omega$, consider the variation
\begin{equation*}
\delta G(\varphi;\xi) := \lim_{\gamma\to 0}\frac{1}{\gamma}\left[G(\varphi +
\gamma\xi) - G(\varphi)\right].   
\end{equation*}
Using (\ref{eq:condF}) and (\ref{eq:partialF}) applied at
$W+\epsilon\int_{0}^{\cdot}\varphi_{s}\ud s$ along with the Dominated
Convergence Theorem, we obtain
\begin{equation*}
\delta G(\varphi;\xi) = \mathbb
E\left[\int_{0}^{T}\varepsilon\Xi_{t}\cdot\partial F\left(W +
\varepsilon\int_{0}^{\cdot}\varphi_{s}\ud s;\ud t\right) -
\int_{0}^{T}\varphi_{t}\cdot\xi_{t}\ud t\right].
\end{equation*}
Using $\Xi=\int_{0}^{\cdot}\xi_{s}\ud s$ and integrating by parts in
the first term on the right-hand-side converts this to
\begin{equation*}
\delta G(\varphi;\xi) = \mathbb
E\left[\int_{0}^{T}\left(\varepsilon\partial F\left(W +
\varepsilon\int_{0}^{\cdot}\varphi_{s}\ud s;(t,T]\right) -
\varphi_{t}\right)\cdot\xi_{t}\ud t\right]. 
\end{equation*} 
The first order condition for the optimal control, $\delta
G(\varphi^{*};\xi)=0$, gives that
\begin{equation*}
\mathbb
E\left[\int_{0}^{T}\varepsilon\partial F\left(W +
\varepsilon\int_{0}^{\cdot}\varphi^{*}_{s}\ud
s;(t,T]\right)\cdot\xi_{t}\ud t\right] = \mathbb E\left[\int_{0}^{T}
\varphi^{*}_{t}\cdot\xi_{t}\ud t\right] 
\end{equation*} 
must hold for every adapted process $\xi$, so (\ref{eq:phistar})
follows. Note that this is the analogue of (\ref{eq:vphistar}) when
performing variational differentiation on Wiener space.

\end{proof}

\begin{remark}
\label{rem:diffF}

If $F$ were Fr\'echet-differentiable (respectively,
Malliavin-differentiable with controls $\varphi$ such that
$\int_{0}^{\cdot}\varphi_{s}\ud s\in\mathcal{CM}$) then the optimiser
would be given by\newline $\varphi^{*}_{t}=\varepsilon\mathbb E[
F^{\prime}\left(W+\varepsilon\int_{0}^{\cdot}\varphi^{*}_{s}\ud
s;(t,T]\right)|\mathcal{F}_{t}]$ (respectively,
$\varphi^{*}_{t}=\varepsilon\mathbb E[
D_{t}F\left(W+\varepsilon\int_{0}^{\cdot}\varphi^{*}_{s}\ud
  s\right))|\mathcal{F}_{t}]$.

\end{remark}

\begin{proof}[Proof of Theorem \ref{thm:asex}]

There are two parts to the proof. First, following the method of
Davis \cite{mhad06}, we use Lemma \ref{lem:ipcan} to approximate
$G(\varphi)$ for small $\varepsilon$, and maximise the approximation
with respect to $\varphi$. We then show that {\em if} one were able to
solve the problem exactly, and then approximate the value function
$G(\varphi^{*})$ for small $\varepsilon$, the same result would
ensue. This will use variational arguments and Proposition
\ref{prop:oc}.

Using Lemma \ref{lem:ipcan} and the martingale representation
(\ref{eq:mrF}) of $F(W)$, the objective functional $G(\varphi)$
in (\ref{eq:Gfunc}) is approximated as
\begin{equation*}
G(\varphi) = \mathbb E\left[F(W) +
\int_{0}^{T}\left(\varepsilon\psi_{t}\cdot\varphi_{t} -
\frac{1}{2}\Vert\varphi_{t}\Vert^{2}\right)\ud t\right] +
o(|\varepsilon|\Vert\varphi\Vert_{\infty}).
\end{equation*}
This is maximised over $\varphi$ by choosing
$\varphi=\widehat{\varphi}:=\varepsilon\psi$, to give
\begin{equation*}
G(\widehat{\varphi}) = \mathbb E[F(W)] + \frac{1}{2}\varepsilon^{2}\mathbb
E\left[\int_{0}^{T}\Vert\psi_{t}\Vert^{2}\ud
t\right] + O(\varepsilon^{4}), 
\end{equation*}
with the remainder term of $O(\varepsilon^{4})$ due to
(\ref{eq:errortermc}). Thus, (\ref{eq:asex}) is indeed obtained by
optimising the approximation to $G(\varphi)$.

For the second part of the proof: using (\ref{eq:asymp2}) along with
(\ref{eq:condF}) and the Dominated Convergence Theorem, we write the
value function $G(\varphi^{*})$ as
\begin{equation}
G(\varphi^{*}) = \mathbb E\left[F(W) + \varepsilon\int_{0}^{T}\partial
F(W;(t,T])\cdot\varphi^{*}_{t}\ud t -
\frac{1}{2}\int_{0}^{T}\Vert\varphi^{*}_{t}\Vert^{2}\ud t\right] +
o(|\varepsilon|\Vert\varphi^{*}\Vert_{\infty}).   
\label{eq:Gphistarexp}
\end{equation}
From (\ref{eq:phistar}), it is evident that (under the mild condition
that $\partial F$ possesses a well-defined first variation) 
\begin{equation}
\varphi^{*}_{t} = \varepsilon\mathbb E[\partial
F(W;(t,T])|\mathcal{F}_{t}] +
o(|\varepsilon|\Vert\varphi^{*}\Vert_{\infty}) = \epsilon\psi_{t} +
o(|\varepsilon|\Vert\varphi^{*}\Vert_{\infty}),
\label{eq:varphiexp}
\end{equation}
the last equality following from the Clark formula
(\ref{eq:clark}). Observe that, to first order in $\varepsilon$,
$\varphi^{*}=\widehat{\varphi}$. 

We now show what would happen if we were to impose a perturbative structure
on the optimal control, that is, if we were to write
\begin{equation}
\varphi^{*}_{t} = \varepsilon\psi_{t} +
\varepsilon^{2}\varphi^{(2)}_{t} + \varepsilon^{3}\varphi^{(3)}_{t} +
O(\varepsilon^{4}),    
\label{eq:vphias}
\end{equation}
for some coefficients $\varphi^{(2)},\varphi^{(3)}$. Supposing
such an expansion were possible, and using this in
(\ref{eq:varphiexp}), we would have
\begin{equation*}
\varepsilon\psi_{t} + \varepsilon^{2}\varphi^{(2)}_{t} +
\varepsilon^{3}\varphi^{(3)}_{t} + O(\varepsilon^{4}) =
\varepsilon\psi_{t} + o(\varepsilon^{2}\Vert\psi\Vert_{\infty}).
\end{equation*}
This would imply, in particular, that $\varphi^{(2)}=0$, and then
(\ref{eq:vphias}) converts to
\begin{equation*}
\varepsilon\varphi^{*}_{t} = \varepsilon^{2}\psi_{t} + O(\varepsilon^{4}).  
\end{equation*}
Using this in (\ref{eq:Gphistarexp}) we obtain
\begin{equation*}
G(\varphi^{*}) = \mathbb E\left[F(W) + \varepsilon^{2}\int_{0}^{T}\partial
F(W;(t,T])\cdot\psi_{t}\ud t -
\frac{1}{2}\varepsilon^{2}\int_{0}^{T}\Vert\psi_{t}\Vert^{2}\ud t\right] +
O(\varepsilon^{4}).   
\end{equation*}
One can use iterated expectations and (\ref{eq:clark}) to convert this
to the statement (\ref{eq:asex}) of the theorem.

\end{proof}

\begin{remark}
\label{rem:ifdiff}

For $F$ sufficiently Fr\'echet differentiable (respectively
  Malliavin differentiable with $\int_{0}^{\cdot}\varphi_{s}\ud
  s\in\mathcal{CM}$), the proof of the asymptotic expansion and the
  quantification of the error term would be more straightforward,
  using a Taylor expansion of
  $F(W+\varepsilon\int_{0}^{\cdot}\varphi_{s}\ud s)$ and of the
  optimal control $\varphi^{*}_{t}=\mathbb
  E[F^{\prime}(W+\varepsilon\int_{0}^{\cdot}\varphi^{*}_{s}\ud
  s;(t,T])|\mathcal{F}_{t}]$ (respectively, $\varphi^{*}_{t}=\mathbb
  E[D_{t}F(W+\varepsilon\int_{0}^{\cdot}\varphi^{*}_{s}\ud
  s;(t,T])|\mathcal{F}_{t}]$).
 
\end{remark}

\section{Dynamic dual representations of indifference price 
processes}
\label{sec:ddrip}

In this section we derive a dynamic dual stochastic control
representation for the exponential indifference price process of a
European claim in a locally bounded semi-martingale market. This will
form the basis for our asymptotic expansion of the indifference
price. Our representation is a slight deviation from the usual way of
expressing the indifference price in terms of relative
entropy. Although the material in this section is mainly classical, we
want a unified treatment that gives dynamic results for unbounded
claims, and this is not readily available in one compact account.

Our approach is to begin with the seminal representation of Grandits
and Rheinl\"ander \cite{gr02} and Kabanov and Stricker \cite{kabstr02}
for an entropy-minimising measure, to establish a dynamic version of
this (Corollary \ref{corr:ZF}), and to use this to establish a dynamic
version (Theorem \ref{thm:fundd}) of the duality result of Delbaen
{\em et al} \cite{6auth}. This result has been obtained for a bounded
claim by Mania and Schweizer \cite{manschw05}. We carry out this
program for a claim satisfying exponential moment conditions akin to
those in Becherer \cite{bech03}. Once we establish duality for the
investment problem with random endowment, we obtain a dynamic version
of the classical dual indifference price representation (Corollary
\ref{cor:dripp}). Then we derive a dynamic result on the entropic
distance between measures (Proposition \ref{prop:ed}) using the
results of \cite{gr02,kabstr02} once more, and this allows us to
convert the classical indifference price representation to our
required representation in Lemma \ref{lem:dripnew}.

The setting is a probability space $(\Omega,\mathcal{F},\mathbb P)$
equipped with a filtration $\mathbb F=(\mathcal{F}_{t})_{0\leq t\leq
  T}$ satisfying the usual conditions of right-continuity and
completeness, where $T\in(0,\infty)$ is a fixed time horizon. We
assume that ${\mathcal{F}}_{0}$ is trivial and that
$\mathcal{F}=\mathcal{F}_{T}$. The discounted prices of $d$ stocks are
modelled by a positive locally bounded semi-martingale $S$. Since we
work with discounted assets, our formulae are unencumbered by any
interest rate adjustments. The class $\mathbf{M}$ of equivalent local
martingale measures (ELMMs) $\mathbb Q$ is of course defined by
\begin{equation*}
\mathbf{M} := \{\mathbb Q\sim\mathbb P|\mbox{$S$ is a $\mathbb Q$-local
martingale}\},    
\end{equation*}
and is assumed non-empty. This assumption is a classical one,
consistent with the absence of arbitrage opportunities, in accordance
with Delbaen and Schachermayer \cite{ds94}.  

Denote by $Z^{\mathbb Q}$ the density process with respect to $\mathbb
P$ of any $\mathbb Q\in\mathbf M$. We write $Z^{\mathbb Q,\mathbb M}$
for the density process of $\mathbb Q\in{\bf M}$ with respect to any
measure $\mathbb M$ other than the physical measure $\mathbb P$,
$\mathbb E^{\mathbb M}$ for expectation with respect to $\mathbb M$,
and $\mathbb E$ for $\mathbb E^{\mathbb P}$.

For $0\leq t\leq T$, we write $Z^{\mathbb Q}_{t,T}:=Z^{\mathbb
  Q}_{T}/Z^{\mathbb Q}_{t}$, with a similar convention for any
positive process. The {\em conditional relative entropy} between
$\mathbb Q\in\mathbf M$ and $\mathbb P$ is the process defined by
\begin{equation*}
I_{t}(\mathbb Q|\mathbb P) := \mathbb E^{\mathbb Q}[\log
Z^{\mathbb Q}_{t,T}|\mathcal{F}_{t}], \quad 0\leq t\leq
T,  
\end{equation*}
provided this is almost surely finite. Define the subset of $\mathbf
M$ given by
\begin{equation*}
\mathbf M_{f} := \{\mathbb Q\in\mathbf M|I_{0}(\mathbb Q|\mathbb P)<\infty\},
\end{equation*}
and we assume throughout that this set of ELMMs with finite relative
entropy is non-empty: $\mathbf M_{f}\neq\emptyset$. By Theorem 2.1 of
Frittelli \cite{fri00}, this implies that there exists a unique
$\mathbb Q^{0}\in\mathbf M_{f}$, the minimal entropy martingale
measure (MEMM), that minimises $I_{0}(\mathbb Q|\mathbb P)$ over all
$\mathbb Q\in\mathbf M_{f}$. It is well-known (for example,
Proposition 4.1 of Kabanov and Stricker \cite{kabstr02}) that the
density process $Z^{\mathbb Q^{0}}$ also minimises the conditional
relative entropy process $I(\mathbb Q|\mathbb P)$ between $\mathbb
Q\in\mathbf M_{f}$ and $\mathbb P$.

The density process of one martingale measure with respect to another
is simply the ratio of their density processes with respect to
$\mathbb P$, as shown in the following lemma.

\begin{lemma}
\label{lem:dpelmm}

Let $\mathbb Q_{1},\mathbb Q_{2}\in\mathbf M_{f}$ have density
processes $Z^{\mathbb Q_{1}},Z^{\mathbb Q_{2}}$ with respect to
$\mathbb P$. Then the density process of $\mathbb Q_{1}$ with respect
to $\mathbb Q_{2}$ is $Z^{\mathbb Q_{1}}/Z^{\mathbb Q_{2}}$.

\end{lemma}

\begin{proof}

Denote by $Z^{\mathbb Q_{1},\mathbb Q_{2}}$ the density process of
$\mathbb Q_{1}$ with respect to $\mathbb Q_{2}$. We have
\begin{equation*}
Z^{\mathbb Q_{1},\mathbb Q_{2}}_{T} := \frac{\ud\mathbb Q_{1}}{\ud\mathbb Q_{2}} =
\frac{\ud\mathbb Q_{1}}{\ud\mathbb P}\left(\frac{\ud\mathbb
Q_{2}}{\ud\mathbb P}\right)^{-1} = \frac{Z^{\mathbb
Q_{1}}_{T}}{Z^{\mathbb Q_{2}}_{T}}.
\end{equation*}
Hence, the $\mathbb Q_{2}$-martingale $Z^{\mathbb Q_{1},\mathbb
Q_{2}}$ is given by
\begin{eqnarray*}
Z^{\mathbb Q_{1},\mathbb Q_{2}}_{t} = \mathbb E^{\mathbb
Q_{2}}[Z^{\mathbb Q_{1},\mathbb Q_{2}}_{T}\vert\mathcal{F}_{t}]
& = & \mathbb E^{\mathbb Q_{2}}\left[\left.\frac{Z^{\mathbb
Q_{1}}_{T}}{Z^{\mathbb Q_{2}}_{T}}\right\vert\mathcal{F}_{t}\right] \\ 
& = & \frac{1}{Z^{\mathbb Q_{2}}_{t}}\mathbb E[Z^{\mathbb
Q_{1}}_{T}|\mathcal{F}_{t}] = \frac{Z^{\mathbb
Q_{1}}_{t}}{Z^{\mathbb Q_{2}}_{t}}, \quad 0\leq t\leq T, 
\end{eqnarray*}
the penultimate equality following from the Bayes rule applied between
$\mathbb Q_{2}$ and $\mathbb P$, and the final equality from the
fact that $Z^{\mathbb Q_{1}}$ is a $\mathbb P$-martingale.
  
\end{proof}

A financial agent trades $S$ and has risk preferences represented by
the exponential utility function
\begin{equation*}
U(x) = - \exp(-\alpha x), \quad \alpha >0, \quad x\in\mathbb R,  
\end{equation*}
with risk aversion coefficient $\alpha$. A European contingent claim
has $\mathcal{F}_{T}$-measurable payoff $F$. Following Becherer
\cite{bech03} and others, we assume that $F$ satisfies suitable
exponential moment conditions:
\begin{equation}
\mathbb E[\exp((\alpha + \varepsilon)F)] < \infty, \quad \mathbb
E[\exp(-\varepsilon F)] <\infty, \quad \mbox{for some $\varepsilon>0$.}
\label{eq:assB} 
\end{equation}
Condition (\ref{eq:assB}) is sufficient to guarantee that $F$ is
$\mathbb Q$-integrable for any $\mathbb Q\in\mathbf M_{f}$ (see for
example Lemma A.1 in Becherer \cite{bech03}).

\subsection{The dynamic primal and dual problems}
\label{subsec:dpdp}

The set $\Theta$ of admissible trading strategies is defined as the
set of $S$-integrable processes $\theta$ such that the stochastic
integral $\theta\cdot S$ is a $\mathbb Q$-martingale for every
$\mathbb Q\in\mathbf M_{f}$, where $\theta$ is a $d$-dimensional
vector representing the number of shares of each stock in the vector
$S$. It is well-known \cite{bech03,6auth,kabstr02,schach01,schach03}
that there are a number of possible choices for a feasible set of
permitted strategies, which all lead to the same value for the dual
problem, defined further below, and it is on this latter problem that
our analysis will be centred. For any $t\in[0,T]$, fix an
$\mathcal{F}_{t}$-measurable random variable $x_{t}$, representing
initial capital. Let $\Theta_{t}$ denote admissible strategies
beginning at $t$.

The primal problem is to maximise expected utility of terminal wealth
generated from trading $S$ and paying the claim payoff at $T$. The
maximal expected utility process is
\begin{equation}
u^{F}_{t}(x_{t}) := \esssup_{\theta\in\Theta_{t}}\mathbb
E\left[\left.-{\rm e}^{-\alpha\left(x_{t} + \int_{t}^{T}\theta_{u}\cdot\ud
  S_{u} - F\right)}\right\vert{\mathcal F}_{t}\right], \quad
0\leq t\leq T,
\label{eq:primal}
\end{equation}
with $\int_{t}^{T}\theta_{u}\cdot\ud
S_{u}=\sum_{i=1}^{d}\int_{t}^{T}\theta^{i}_{u}\ud S^{i}_{u}$.

We shall use the notational convention whereby setting $F=0$ in
(\ref{eq:primal}) signifies the corresponding quantity in the problem
without the claim. Hence, the classical investment problem without the
claim has maximal expected utility process $u^{0}$. Denote the
optimiser in (\ref{eq:primal}) by $\theta^{F}$, so $\theta^{0}$ is the
optimiser in the problem without the claim.

The utility indifference price process for the claim, $p(\alpha)$, is
defined by
\begin{equation*}
u_{t}^{F}(x_{t} + p_{t}(\alpha)) = u_{t}^{0}(x_{t}), \quad 0\leq t\leq
T.  
\end{equation*}
It is well-known (see for instance Becherer \cite{bech03} or Mania and
Schweizer \cite{manschw05}) that, with exponential utility,
$p(\alpha)$ has no dependence on the starting capital (this follows
from (\ref{eq:primal}), where the initial capital factors out of the
optimisation). The hedging strategy associated with this pricing
mechanism is $\theta(\alpha)$, defined by
\begin{equation*}
\theta(\alpha) := \theta^{F} - \theta^{0}.  
\end{equation*}

The dual problem to (\ref{eq:primal}) is defined by
\begin{equation}
I^{F}_{t} := \essinf_{\mathbb Q\in\mathbf M_{f}}\left[I_{t}(\mathbb
Q|\mathbb P) - \alpha\mathbb E^{\mathbb Q}[F|\mathcal{F}_{t}]\right],
\quad 0\leq t\leq T.   
\label{eq:dual}
\end{equation}
Denote the optimiser in (\ref{eq:dual}) by $\mathbb Q^{F}$, so the
optimiser without the claim is $\mathbb Q^{0}$, the MEMM.

It is well-known (at least in a static context) that if we define the
measure $\mathbb P_{F}\sim\mathbb P$ by
\begin{equation}
\frac{\ud\mathbb P_{F}}{\ud\mathbb P} := \frac{\exp(\alpha F)}{\mathbb
E[\exp(\alpha F)]},  \label{eq:PB}
\end{equation}
then we can use $\mathbb P_{F}$ instead of $\mathbb P$ as our
reference measure, and this removes the claim from the primal and dual
problems. In the dual picture, therefore, $\mathbb Q^{F}$ is the
martingale measure which minimises the relative entropy between any
$\mathbb Q\in\mathbf M_{f}$ and $\mathbb P_{F}$. These properties of
$\mathbb P_{F}$ are well-known in a static context from Delbaen {\em
  et al} \cite{6auth}. The dynamic analogue of these arguments is
given below.

Note that if we use $\mathbb P_{F}$ instead of $\mathbb P$ as
reference measure, one could (in principle) define a set $\mathbf
M_{f}(\mathbb P_{F})$ of ELMMs with finite relative entropy with
respect to $\mathbb P_{F}$, but it is well-known that $\mathbf
M_{f}(\mathbb P_{F})=\mathbf M_{f}(\mathbb P)$ (see the statement and
proof of Lemma A.1 in Becherer \cite{bech03}, for example) so we
simply write $\mathbf M_{f}$.

Define the $\mathbb P$-martingale $M^{F}$ as the density process of
$\mathbb P_{F}$ with respect to $\mathbb P$:
\begin{equation*}
M^{F}_{t} := \left.\frac{\ud\mathbb P_{F}}{\ud\mathbb
  P}\right\vert_{\mathcal{F}_{t}} = \mathbb
E\left[\left.\frac{\ud\mathbb P_{F}}{\ud\mathbb
    P}\right\vert\mathcal{F}_{t}\right] = \frac{\mathbb E[e^{\alpha
      F}|\mathcal{F}_{t}]}{\mathbb E[e^{\alpha F}]}, \quad 0\leq t\leq T,
\end{equation*}
which satisfies, for any integrable $\mathcal{F}_{T}$-measurable
random variable $V$,
\begin{equation}
\mathbb E^{\mathbb P_{F}}[V|\mathcal{F}_{t}] =
\frac{1}{M^{F}_{t}}\mathbb E[M^{F}_{T}V|\mathcal{F}_{t}], \quad 0\leq
t\leq T.
\label{eq:MFV}
\end{equation}
We ``remove the claim'' from the primal problem using the measure
$\mathbb P_{F}$ as follows. Using (\ref{eq:MFV}) we convert
(\ref{eq:primal}) to
\begin{equation*}
u^{F}_{t}(x_{t}) := \mathbb E[e^{\alpha
F}|\mathcal{F}_{t}]\esssup_{\theta\in\Theta_{t}}\mathbb E^{\mathbb
P_{F}}\left[\left.-{\rm e}^{-\alpha\left(x_{t} +
\int_{t}^{T}\theta_{u}\cdot\ud S_{u}\right)}\right\vert{\mathcal
F}_{t}\right], \quad 0\leq t\leq T,
\end{equation*}
from which it is apparent that one may optimise over strategies in a
problem without the claim and with $\mathbb P_{F}$ as reference
measure. The same approach also works, of course, for the dual
problem, as we show below. We shall need the following simple result
relating the density process of any $\mathbb Q\in\mathbf M_{f}$ with
respect to $\mathbb P$ to its counterpart with respect to $\mathbb
P_{F}$.

\begin{lemma}
\label{lem:ZQFZQ}

For any $\mathbb Q\in\mathbf M_{f}$, the density processes $Z^{\mathbb
  Q}$ and $Z^{\mathbb Q,\mathbb P_{F}}$ are related by
\begin{equation*}
Z^{\mathbb Q}_{t} = M^{F}_{t}Z^{\mathbb Q,\mathbb P_{F}}_{t}, \quad 0\leq
t\leq T.  
\end{equation*}
  
\end{lemma}

\begin{proof}
  
For $\mathbb Q\in\mathbf M_{f}$, we have
\begin{eqnarray*}
Z^{\mathbb Q,\mathbb P_{F}}_{t} = \mathbb E^{\mathbb
P_{F}}\left[\left.\frac{\ud\mathbb Q}{\ud\mathbb
P_{F}}\right\vert\mathcal{F}_{t}\right] & = & \mathbb E^{\mathbb
P_{F}}\left[\left.\frac{\ud\mathbb Q}{\ud\mathbb
P}\bigg/\frac{\ud\mathbb P_{F}}{\ud\mathbb
P}\right\vert\mathcal{F}_{t}\right] \\
& = & \mathbb E^{\mathbb
P_{F}}\left[\left.\frac{1}{M^{F}_{T}}\frac{\ud\mathbb Q}{\ud\mathbb
P}\right\vert\mathcal{F}_{t}\right] \\
& = & \frac{1}{M^{F}_{t}}\mathbb
E\left[\left.\frac{\ud\mathbb Q}{\ud\mathbb
P}\right\vert\mathcal{F}_{t}\right] \\
& = & \frac{Z^{\mathbb Q}_{t}}{M^{F}_{t}}, \quad 0\leq t\leq T,
\end{eqnarray*}
where we have used (\ref{eq:MFV}). 

\end{proof}

Applying Lemma \ref{lem:ZQFZQ} in turn at $t\leq T$ and at $T$, we obtain
\begin{equation}
Z^{\mathbb Q,\mathbb P_{F}}_{t,T} = \frac{Z^{\mathbb
Q}_{t,T}}{M^{F}_{t,T}} = \frac{\mathbb E[{\rm e}^{\alpha
F}|\mathcal{F}_{t}]}{{\rm e}^{\alpha F}}Z^{\mathbb Q}_{t,T}, \quad 0\leq
t\leq T.  
\label{eq:ZFtT}
\end{equation}
We use this to ``remove the claim'' from the dual problem
(\ref{eq:dual}): compute, for any $\mathbb Q\in\mathbf
M_{f}$,
\begin{eqnarray*}
I_{t}(\mathbb Q|\mathbb P_{F}) & = & \mathbb E^{\mathbb Q}[\log
Z^{\mathbb Q,\mathbb P_{F}}_{t,T}|\mathcal{F}_{t}] \\
& = & I_{t}(\mathbb Q|\mathbb P) - \alpha\mathbb E^{\mathbb
Q}[F|\mathcal{F}_{t}] + \log(\mathbb E[{\rm e}^{\alpha
F}|\mathcal{F}_{t}]), \quad 0\leq t\leq T.  
\end{eqnarray*}
Using this in (\ref{eq:dual}), we obtain
\begin{equation}
I^{F}_{t} = \essinf_{\mathbb Q\in\mathbf M_{f}}[I_{t}(\mathbb
Q|\mathbb P_{F})] - \log(\mathbb E[{\rm e}^{\alpha
F}|\mathcal{F}_{t}]), \quad 0\leq t\leq T.  
\label{eq:IFI}
\end{equation}
Since the last term on the right-hand-side does not depend on $\mathbb
Q$, we see that we can reduce the dual problem to the problem
\begin{equation*}
I_{t}(\mathbb Q|\mathbb P_{F}) \longrightarrow \min !,  
\end{equation*}
so that $\mathbb Q^{F}$ minimises $I(\mathbb Q|\mathbb P_{F})$, and,
when $F=0$, $\mathbb Q^{0}$ is the MEMM.

\subsection{The fundamental duality}
\label{subsec:fd}

The duality results we need follow from the representation below for
$Z^{\mathbb Q^{F},\mathbb P_{F}}$, originally proven independently (to
the best of our knowledge) by Grandits and Rheinl\"ander \cite{gr02}
and Kabanov and Stricker \cite{kabstr02} for $F=0$ (and hence for
$Z^{\mathbb Q^{0}}$), but which applies equally well to $\mathbb
Q^{F}$ if we use $\mathbb P_{F}$ as reference measure. Both
\cite{gr02} and \cite{kabstr02} prove the result for a market
involving a locally bounded semi-martingale $S$. This has been
generalised to a general semi-martingale by Biagini and Frittelli
\cite{bf07}.

\begin{property}[\cite{gr02,kabstr02}]
\label{property:gr}

The density of the dual minimiser $\mathbb Q^{F}$ in (\ref{eq:dual})
with respect to the measure $\mathbb P_{F}$ defined in (\ref{eq:PB})
is given by
\begin{equation}
\frac{\ud\mathbb Q^{F}}{\ud\mathbb P_{F}} \equiv Z^{\mathbb
  Q^{F},\mathbb P_{F}}_{T} = c_{F}\exp(-\alpha(\theta^{F}\cdot
S)_{T}), \quad c_{F}\in\mathbb R_{+},  
\label{eq:grF}
\end{equation}
where $\theta^{F}\in\Theta$ is the optimal strategy in the primal
problem (\ref{eq:primal}) and the stochastic integral
$(\theta^{F}\cdot S)$ is a $\mathbb Q$-martingale for any $\mathbb
Q\in\mathbf M_{f}$.

\end{property}

We convert this to the dynamic result below, in which we also restore
$\mathbb P$ as reference measure.

\begin{corollary}
\label{corr:ZF}

The density process $Z^{\mathbb Q^{F}}$ of the dual minimiser $\mathbb
Q^{F}$ in (\ref{eq:dual}) satisfies, for $t\in[0,T]$, 
\begin{equation}
Z^{\mathbb Q^{F}}_{t,T} = \exp\left[I_{t}(\mathbb Q^{F}|\mathbb P) -
\alpha\left(\mathbb E^{\mathbb Q^{F}}[F|\mathcal{F}_{t}] +
\int_{t}^{T}\theta^{F}_{u}\cdot\ud S_{u} - F\right)\right], 
\label{eq:ZQF}
\end{equation}
where $\theta^{F}\in\Theta$ is the optimal strategy in the primal
problem (\ref{eq:primal}).
  
\end{corollary}

\begin{proof}

First, we obtain a dynamic version of (\ref{eq:grF}). Using
(\ref{eq:grF}) and the $\mathbb Q^{F}$-martingale property of
$(\theta^{F}\cdot S)$, we have
\begin{eqnarray*}
I_{t}(\mathbb Q^{F}|\mathbb P_{F}) & = & \mathbb E^{\mathbb Q^{F}}[\log
Z^{\mathbb Q^{F},\mathbb P_{F}}_{t,T}|\mathcal{F}_{t}] \\  
& = & \mathbb E^{\mathbb Q^{F}}[\log c_{F} - \alpha(\theta^{F}\cdot
S)_{T}|\mathcal{F}_{t}] - \log Z^{\mathbb Q^{F},\mathbb P_{F}}_{t} \\
& = & \log c_{F} - \alpha(\theta^{F}\cdot S)_{t} - \log Z^{\mathbb
Q^{F},\mathbb P_{F}}_{t}, \quad 0\leq t\leq T.
\end{eqnarray*}
Using this in turn at $t\leq T$ and $T$ we obtain
\begin{equation*}
Z^{\mathbb Q^{F},\mathbb P_{F}}_{t,T} =
c^{F}_{t}\exp\left(-\alpha\int_{t}^{T}\theta^{F}_{u}\cdot\ud
S_{u}\right), \quad c^{F}_{t} := \exp(I_{t}(\mathbb Q^{F}|\mathbb P_{F}),
  \quad 0\leq t\leq T,   
\end{equation*}
which is a dynamic version of (\ref{eq:grF}). Using this along with
(\ref{eq:ZFtT}) and (\ref{eq:IFI}) we obtain
\begin{equation*}
Z^{\mathbb Q^{F}}_{t,T} = \exp\left(I^{F}_{t} -
\alpha\int_{t}^{T}\theta^{F}_{u}\cdot\ud S_{u} + \alpha F\right),
\quad 0\leq t\leq T.
\end{equation*}
Finally, using the definition (\ref{eq:dual}) of $I^{F}$ gives the
result. 
  
\end{proof}

Corollary \ref{corr:ZF} is nothing more than a dynamic version of the
classical result of Grandits and Rheinl\"ander \cite{gr02} and Kabanov
and Stricker \cite{kabstr02} for the MEMM, with the added
generalisation of allowing for $\mathbb P_{F}$ as reference
measure. It leads immediately to the duality result below, a dynamic
version of the duality in Delbaen {\em et al} \cite{6auth}. This
result is stated in Mania and Schweizer \cite{manschw05} for a bounded
claim. We give a proof to highlight that the boundedness condition on
the claim is not needed.

\begin{theorem}[\cite{6auth,bech03,kabstr02,manschw05}]
\label{thm:fundd}

Suppose the claim payoff $F$ satisfies the exponential moment
conditions (\ref{eq:assB}). Then the maximal expected utility process
in (\ref{eq:primal}) and the optimal dual process in (\ref{eq:dual})
are related by
\begin{equation}
u^{F}_{t}(x_{t}) = -\exp\left(-\alpha x_{t} - I^{F}_{t}\right),
\quad 0\leq t\leq T. \label{eq:fundd}
\end{equation}

\end{theorem}

\begin{proof}

We compute the primal optimal expected utility process and use
Corollary \ref{corr:ZF} to substitute for the stochastic integral
$(\theta^{F}\cdot S)$: 
\begin{eqnarray*}
u^{F}_{t}(x_{t}) & = & \mathbb E\left[\left.-{\rm e}^{-\alpha\left(x_{t} +
\int_{t}^{T}\theta^{F}_{u}\cdot\ud S_{u} -
F\right)}\right\vert{\mathcal F}_{t}\right] \\
& = & -{\rm e}^{-\alpha x_{t}}\mathbb E[Z^{\mathbb
Q^{F}}_{t,T}\exp(-I^{F}_{t})|\mathcal{F}_{t}] \quad \mbox{(using
Corollary   \ref{corr:ZF})} \\
& = & -\exp\left(-\alpha x_{t} - I^{F}_{t}\right), \quad 0\leq t\leq T. 
\end{eqnarray*}

\end{proof}

Using this theorem and the definition of the indifference price we
obtain the following dual representation of the indifference price
process, a dynamic version of the classical representation.

\begin{corollary}
\label{cor:dripp}
  
The indifference price process has the dual representation
\begin{equation}
p_{t}(\alpha) = -\frac{1}{\alpha}(I^{F}_{t} - I^{0}_{t}),
\quad 0\leq t\leq T. 
\label{eq:erip}
\end{equation}

\end{corollary}

Written out explicitly, (\ref{eq:erip}) can be re-cast into the
more familiar form
\begin{equation}
p_{t}(\alpha) = \esssup_{\mathbb Q\in\mathbf M_{f}}\left[\mathbb
E^{\mathbb Q}[F|{\mathcal F}_{t}] - \frac{1}{\alpha}\left(I_{t}(\mathbb
Q|\mathbb P) - I_{t}(\mathbb Q^{0}|\mathbb P)\right)\right], \quad
0\leq t\leq T. 
\label{eq:erip2} 
\end{equation}
The two conditional entropy terms in (\ref{eq:erip2}) can in fact be
condensed into one, using the following proposition.

\begin{proposition}
\label{prop:ed}

The conditional entropy process $I$ satisfies the property that, for
any equivalent local martingale measure $\mathbb Q\in\mathbf M_{f}$, 
\begin{equation}
I_{t}(\mathbb Q|\mathbb P) - I_{t}(\mathbb Q^{0}|\mathbb P) = I_{t}(\mathbb
Q|\mathbb Q^{0}), \quad 0\leq t\leq T. 
\label{eq:IQQ0}
\end{equation}
  
\end{proposition}

\begin{proof}

For any $\mathbb Q\in\mathbf M_{f}$, the conditional
entropy process $I(\mathbb Q|\mathbb Q^{0})$ is given by
\begin{eqnarray}
I_{t}(\mathbb Q|\mathbb Q^{0}) & := & \mathbb E^{\mathbb Q}[\log Z^{\mathbb
Q,\mathbb Q^{0}}_{t,T}|{\mathcal F}_{t}] \nonumber \\
& = & \mathbb E^{\mathbb Q}[\log Z^{\mathbb
Q}_{t,T} - \log Z^{\mathbb Q^{0}}_{t,T}|{\mathcal F}_{t}] \nonumber  \\
& = & I_{t}(\mathbb Q|\mathbb P) - \mathbb E^{\mathbb Q}[\log Z^{\mathbb
Q^{0}}_{t,T}|{\mathcal F}_{t}], \quad 0\leq t\leq T.  
\label{eq:HQQ0}
\end{eqnarray}
We have the dynamic version of the Grandits-Rheinl\"ander
\cite{gr02} representation of the MEMM, given by (\ref{eq:ZQF}) for
$F=0$: 
\begin{equation*}
Z^{\mathbb Q^{0}}_{t,T} = \exp\left(I_{t}(\mathbb Q^{0}|\mathbb P) -
\alpha\int_{t}^{T}\theta^{0}_{u}\cdot\ud S_{u}\right), \quad 0\leq
t\leq T,  
\end{equation*}
where the optimal investment strategy $\theta^{0}\in\Theta$, so
$(\theta^{0}\cdot S)$ is a $\mathbb Q$-martingale, for any $\mathbb
Q\in\mathbf M_{f}$. Using this in (\ref{eq:HQQ0}) we obtain
(\ref{eq:IQQ0}). 
  
\end{proof}

Using Proposition \ref{prop:ed} in the classical dual stochastic
control representation (\ref{eq:erip}) of the indifference price
process, we immediately obtain the following form for $p(\alpha)$,
which will form the basis for our asymptotic expansion of the
indifference price process.

\begin{lemma}
\label{lem:dripnew}

The indifference price process is given by the dual stochastic control
representation 
\begin{equation*}
p_{t}(\alpha) = \esssup_{\mathbb Q\in\mathbf M_{f}}\left[\mathbb
E^{\mathbb Q}[F|{\mathcal F}_{t}] - \frac{1}{\alpha}I_{t}(\mathbb
Q|\mathbb Q^{0})\right], \quad 0\leq t\leq T. 
\end{equation*}

\end{lemma}

\begin{proof}

Use (\ref{eq:IQQ0}) in (\ref{eq:erip}). 

\end{proof}

\begin{remark}
\label{rem:lsz}

A version of Lemma \ref{lem:dripnew} for American claims was given in
Leung {\em et al} \cite{lsz12} in a stochastic volatility scenario
(see their Proposition 7).

\end{remark}

\begin{remark}
\label{rem:oqf}

The optimiser in Lemma \ref{lem:dripnew} is also the optimiser in
(\ref{eq:erip2}), that is, $\mathbb Q^{F}$. 

\end{remark}

\section{Indifference valuation in an incomplete It\^o process 
market} 
\label{sec:iviim}

In this section we apply the indifference pricing formula from Lemma
\ref{lem:dripnew} in an It\^o process setting, and we show how it
leads to a control problem of a similar structure to the one analysed
in Section \ref{sec:cp}.

We have a probability space $(\Omega,{\mathcal F},\mathbb P)$ equipped
with the standard augmented filtration $\mathbb F:=({\mathcal
  F}_{t})_{0\leq t\leq T}$ associated with an $m$-dimensional Brownian
motion $W$. On this space we have a financial market with (for
simplicity) zero interest rate. The price processes of $d<m$ stocks
are given by the vector $S=(S^{1},\ldots,S^{d})^{\top}$, where
$S=(S_{t})_{0\leq t\leq T}$ follows the It\^o process
\begin{equation}
\ud S_{t} = \diag_{d}(S_{t})[\mu^{S}_{t}\ud t + \sigma_{t}\ud W_{t}],
\label{eq:sdyn}
\end{equation}
with $\diag_{d}(\cdot)$ denoting the $(d\times d)$ matrix with zero
entries off the main diagonal. The $d$-dimensional appreciation rate
vector $\mu^{S}$ and the $(d\times m)$ volatility matrix $\sigma$ are
$\mathbb F$-progressively measurable processes satisfying
$\int_{0}^{T}\Vert\mu^{S}_{t}\Vert\ud t<\infty$ and
$\int_{0}^{T}\sigma_{t}\sigma^{\top}_{t}\ud t<\infty$, almost surely. The
volatility matrix $\sigma_{t}$ has full rank for every $t\in[0,T]$, so
that the matrix $(\sigma_{t}\sigma^{\top}_{t})^{-1}$ is well-defined, as
is the $m$-dimensional relative risk process given by
\begin{equation}
\lambda_{t} := \sigma^{\top}_{t}(\sigma_{t}\sigma^{\top}_{t})^{-1}\mu^{S}_{t},
\quad 0\leq t\leq T. 
\label{eq:lambda}
\end{equation}

For $d<m$, this market is incomplete. We also have a vector
$Y=(Y^{1},\ldots,Y^{m-d})^{\top}$ of $(m-d)$ non-traded factors. These
could be the prices of non-traded assets, or of factors such as
stochastic volatilities and correlations. This framework is general
enough to encompass multi-dimensional versions of basis risk models as
well as multi-factor stochastic volatility models, with no Markovian
structure needed. We assume that $Y$ follows the It\^o process
\begin{equation*}
\ud Y_{t} = \diag_{m-d}(Y_{t})[\mu^{Y}_{t}\ud t + \beta_{t}\ud W_{t}],
\end{equation*}
for an $(m-d)$-dimensional progressively measurable vector $\mu^{Y}$
satisfying $\int_{0}^{T}\Vert\mu^{Y}_{t}\Vert\ud t<\infty$, almost surely,
and an $(m-d)\times m$-dimensional progressively measurable matrix
$\beta$ satisfying $\int_{0}^{T}\beta_{t}\beta^{\top}_{t}\ud
t<\infty$, almost surely.

A European contingent claim has $\mathcal{F}_{T}$-measurable payoff
$F$ depending on the evolution of $(S,Y)$. We assume $F$ satisfies
Assumption \ref{ass:condF}, so in particular, $F\in L^{2}(\mathbb Q)$,
for any ELMM $\mathbb Q\in\mathbf{M}_{f}$.

Measures $\mathbb Q\sim\mathbb P$ have density processes with
respect to $\mathbb P$ of the form
\begin{equation}
Z^{\mathbb Q}_{t} = {\mathcal E}(-q\cdot W)_{t}, \quad 0\leq t\leq
T,
\label{eq:ZmbbQ}
\end{equation}
for some $m$-dimensional process $q$ satisfying $\int_{0}^{T}\Vert
q_{t}\Vert^{2}\ud t<\infty$ almost surely. For $Z^{\mathbb Q}$ to be
the density of an equivalent local martingale measure, it must be is a
$\mathbb P$-martingale (a Novikov condition on $q$ would guarantee
this) and in addition $q$ must satisfy
\begin{equation}
\mu^{S}_{t} - \sigma_{t}q_{t} = {\bf 0}_{d}, \quad 0\leq t\leq
T,  
\label{eq:mcphi} 
\end{equation}
where ${\bf 0}_{d}$ denotes the $d$-dimensional zero vector, so that
$S$ is a local $\mathbb Q$-martingale. 

As the market is incomplete, there will be an infinite number of
solutions $q$ to the equations (\ref{eq:mcphi}), and the ELMMs
$\mathbb Q$ are in one-to-one correspondence with processes $q$
satisfying (\ref{eq:mcphi}) and such that ${\mathcal E}(-q\cdot W)$ is
a $\mathbb P$-martingale.

By the Girsanov theorem, the process $W^{\mathbb Q}$ defined by
\begin{equation}
W^{\mathbb Q}_{t} := W_{t} + \int_{0}^{t}q_{u}\ud u, \quad 0\leq
t\leq T,
\label{eq:WmbbQ}
\end{equation}
is an $m$-dimensional $\mathbb Q$-Brownian motion. The dynamics of the
stocks and non-traded factors under $\mathbb Q$ are then
\begin{eqnarray}
\ud S_{t} & = & \diag_{d}(S_{t})\sigma_{t}\ud W^{\mathbb
Q}_{t}, \label{eq:sqdyn} \\ 
\ud Y_{t} & = & \diag_{m-d}(Y_{t})[(\mu^{Y}_{t} - \beta_{t}q_{t})\ud t +
\beta_{t}\ud W^{\mathbb Q}_{t}]. \label{eq:yqdyn}
\end{eqnarray}
If we choose $q=\lambda$, given by (\ref{eq:lambda}), we obtain the
minimal martingale measure $\mathbb Q_{M}$, while the density process
of the MEMM $\mathbb Q^{0}$ is $Z^{\mathbb
  Q^{0}}=\mathcal{E}(-q^{0}\cdot W)$, for some integrand $q^{0}$.

Denote by $H^{2}(\mathbb Q)$ the space of $L^{2}$-bounded continuous
$\mathbb Q$-martingales $M$ (so, $\sup_{t\in[0,T]}\mathbb E^{\mathbb
  Q}[M^{2}_{t}]<\infty$). By Proposition IV.1.23 and Corollary IV.1.25
in Revuz and Yor \cite{ry99}, $H^{2}(\mathbb Q)$ is also the space of
martingales $M$ such that $\mathbb E^{\mathbb Q}[[M]_{T}]<\infty$.
Denoting $\Lambda^{\mathbb Q}:=(q\cdot W^{\mathbb Q})$, then using
(\ref{eq:ZmbbQ}) and (\ref{eq:WmbbQ}), $\log Z^{\mathbb
  Q}=-\Lambda^{\mathbb Q}+[\Lambda^{\mathbb Q}]/2$, so the relative
entropy between $\mathbb Q\in\mathbf M_{f}$ and $\mathbb P$ is given
by
\begin{equation*}
0\leq I_{0}(\mathbb Q|\mathbb P) = \mathbb E^{\mathbb
Q}\left[-\Lambda^{\mathbb Q}_{T} + \frac{1}{2}[\Lambda^{\mathbb Q}]_{T}\right] <
\infty,  
\end{equation*}
the last inequality true by assumption. The finiteness and
non-negativity of this relative entropy yields that both expectations
above are finite. Precisely, we have $\mathbb E^{\mathbb
  Q}[\Lambda^{\mathbb Q}_{T}]>-\infty$ and, in particular, $\mathbb
E^{\mathbb Q}[[\Lambda^{\mathbb Q}]_{T}]<\infty$, the latter condition
implying that $\Lambda^{\mathbb Q}\in H^{2}(\mathbb Q)$. Therefore,
\begin{equation}
\Lambda^{\mathbb Q} := (q\cdot W^{\mathbb Q}) \quad \mbox{is a
  $\mathbb Q$-martingale, for all $\mathbb Q\in\mathbf{M}_{f}$}.
\label{eq:Lmgl}
\end{equation}
This will be useful in computing the conditional relative
entropy $I(\mathbb Q|\mathbb Q^{0})$. 

Using (\ref{eq:WmbbQ}) in turn for $\mathbb Q$ and $\mathbb Q^{0}$, we
have
\begin{equation}
W^{\mathbb Q}_{t} = W^{\mathbb Q^{0}}_{t} + \int_{0}^{t}(q_{t} -
q^{0}_{t})\ud t, \quad 0\leq t\leq T,
\label{eq:wqq0}
\end{equation}
where $W^{\mathbb Q^{0}}$ is a $\mathbb Q^{0}$-Brownian motion.
 
Note that since both $q$ and $q^{0}$ satisfy (\ref{eq:mcphi}), we have
\begin{equation}
\sigma_{t}(q_{t} - q^{0}_{t}) = \mathbf{0}_{d}, \quad 0\leq t\leq T,
\label{eq:mcq}
\end{equation}
which we shall use later.

Using (\ref{eq:wqq0}), we can write the $\mathbb Q$-dynamics of $Y$ in
(\ref{eq:yqdyn}) as
\begin{equation*}
\ud Y_{t} = \diag_{m-d}(Y_{t})[(\mu^{Y}_{t} - \beta_{t}q^{0}_{t})\ud t +
\beta_{t}(\ud W^{\mathbb Q}_{t} - (q_{t} - q^{0}_{t})\ud t)].
\end{equation*}
The point of this representation is that the $\mathbb Q$-dynamics of
$Y$ may be interpreted as a perturbation of the $\mathbb
Q^{0}$-dynamics, since setting $q=q^{0}$ gives the dynamics under the
MEMM $\mathbb Q^{0}$, with the Brownian motion $W^{\mathbb Q}$ also
being modulated by the choice of $q$.

Using (\ref{eq:WmbbQ}) and (\ref{eq:wqq0}), the density process of
$\mathbb Q$ with respect to $\mathbb Q^{0}$ is
\begin{equation*}
Z^{\mathbb Q,\mathbb Q^{0}}_{t} = \frac{Z^{\mathbb Q}_{t}}{Z^{\mathbb
Q^{0}}_{t}} = \frac{\mathcal{E}(-q\cdot
  W)_{t}}{\mathcal{E}(-q^{0}\cdot W)_{t}} = \mathcal{E}(-(q -
q^{0})\cdot W^{\mathbb Q^{0}})_{t}, \quad 0\leq t\leq T.
\end{equation*}
Using this, along with (\ref{eq:wqq0}) and the martingale condition
(\ref{eq:Lmgl}), we compute
\begin{equation}
I_{t}(\mathbb Q|\mathbb Q^{0}) = \mathbb E^{\mathbb
Q}\left[\left.\frac{1}{2}\int_{t}^{T}\Vert q_{u}-q^{0}_{u}\Vert^{2}\ud 
u\right\vert{\mathcal F}_{t}\right], \quad 0\leq t\leq T.
\label{eq:iqq0}
\end{equation}
Now we explicitly consider $\mathbb Q$ as a perturbation around
$\mathbb Q^{0}$. Introduce, for some small parameter $\varepsilon$, a
parametrised family of measures $\{\mathbb
Q(\varepsilon)\}_{\varepsilon\in\mathbb R}$, such that
\begin{equation}
\mathbb Q \equiv \mathbb Q(\varepsilon), \quad \mathbb Q^{0} \equiv
\mathbb Q(0), 
\label{eq:qeq0}
\end{equation}
and also write
\begin{equation}
q - q^{0} =: -\varepsilon\varphi,  \label{eq:psiphi}
\end{equation}
for some process $\varphi$. Then (\ref{eq:mcq}) becomes
\begin{equation}
\sigma\varphi = {\bf 0}_{d}.
\label{eq:cphi}  
\end{equation}
Denote by $\mathcal{A}(\mathbf M_{f})$ the set of such $\varphi$ which
correspond to $\mathbb Q\in\mathbf M_{f}$, and also define the process
$\Phi:=\int_{0}^{\cdot}\varphi_{s}\ud s$.

The $\mathbb Q(\varepsilon)$-dynamics of the state variables $S,Y$ in
this notation are then
\begin{eqnarray}
\ud S_{t} & = & \diag_{d}(S_{t})\sigma_{t}\ud W^{\mathbb
Q(\varepsilon)}_{t}, \label{eq:sedyn} \\ 
\ud Y_{t} & = & \diag_{m-d}(Y_{t})[(\mu^{Y}_{t} -
\beta_{t}q^{0}_{t})\ud t + \beta_{t}(\ud W^{\mathbb
  Q(\varepsilon)}_{t} + \varepsilon\varphi_{t}\ud t)]. \label{eq:yedyn} 
\end{eqnarray}
Observe that if we define the state variable $X:=(S,Y)^{\top}$, then we
have recovered dynamics of the general form (\ref{eq:xeps2}). 

The $\mathbb Q(\varepsilon)$-dynamics (\ref{eq:sedyn}) of $S$, along
with the constraint (\ref{eq:cphi}), lead to the following
orthogonality result between trading strategies and dual controls.
Consider integrands $\theta^{(\varepsilon)},\varphi$ such that
$(\theta^{(\varepsilon)}\cdot S)$ is a $\mathbb
Q(\varepsilon)$-martingale and $\varphi$ satisfies
(\ref{eq:cphi}). Then a straightforward computation using
(\ref{eq:sedyn}) and (\ref{eq:cphi}) shows that the stochastic
integrals $(\theta^{(\varepsilon)}\cdot S)$ and $(\varphi\cdot
W^{\mathbb Q(\varepsilon)})$ are orthogonal $\mathbb
Q(\varepsilon)$-martingales. That is, $\mathbb E^{\mathbb
  Q(\varepsilon)}[(\theta^{(\varepsilon)}\cdot S)_{T}(\varphi\cdot
W^{\mathbb Q(\varepsilon)})_{T}] = 0$.  In particular, this will hold
for $\varepsilon=0$.

A similar orthogonality result is reflected in the following
decomposition of the claim payoff $F$. When the dynamics of the state
variables are given as in (\ref{eq:sedyn}) and (\ref{eq:yedyn}), we
write $F\equiv F(W^{\mathbb Q(\varepsilon)}+\varepsilon\Phi)$. Write
the Galtchouk-Kunita-Watanabe decomposition of $F$ under $\mathbb
Q(0)\equiv\mathbb Q^{0}$ as
\begin{equation}
F(W^{\mathbb Q(0)}) = \mathbb E^{\mathbb Q(0)}[F(W^{\mathbb Q(0)})] +
(\theta^{(0)}\cdot S)_{T} + (\xi^{(0)}\cdot W^{\mathbb Q(0)})_{T}, 
\label{eq:gkwde}
\end{equation}
for some integrands $\theta^{(0)},\xi^{(0)}$, such that the
stochastic integrals in (\ref{eq:gkwde}) are orthogonal $\mathbb
Q(0)$-martingales, so we have
\begin{equation*}
\mathbb E^{\mathbb Q(0)}[(\theta^{(0)}\cdot
S)_{T}(\xi^{(0)}\cdot W^{\mathbb Q(0)})_{T}] = 0.
\end{equation*}

Using (\ref{eq:iqq0}) and (\ref{eq:psiphi}), the indifference price
process, as given by Lemma \ref{lem:dripnew}, has the stochastic
control representation
\begin{equation*}
p_{t}(\alpha) = \sup_{\varphi\in\mathcal{A}(\mathbf M_{f})}\mathbb E^{\mathbb
Q(\varepsilon)}\left[\left.F(W^{\mathbb Q(\varepsilon)} + \varepsilon\Phi) -
\frac{\varepsilon^{2}}{2\alpha}\int_{t}^{T}\Vert\varphi_{u}\Vert^{2}\ud
u\right\vert{\mathcal F}_{t}\right], \quad 0\leq t\leq T.   
\end{equation*}
If we choose 
\begin{equation}
\varepsilon^{2} = \alpha,
\label{eq:epsalpha}
\end{equation}
then we get a control problem of the form
\begin{equation*}
p_{t}(\alpha) = \sup_{\varphi\in\mathcal{A}(\mathbf M_{f})}\mathbb E^{\mathbb
Q(\varepsilon)}\left[\left.F\left(W^{\mathbb Q(\varepsilon)} +
\varepsilon\int_{t}^{\cdot}\varphi_{u}\ud u\right) -
\frac{1}{2}\int_{t}^{T}\Vert\varphi_{u}\Vert^{2}\ud u\right\vert{\mathcal
F}_{t}\right], \quad 0\leq t\leq T. 
\end{equation*}
subject to $\mathbb Q(\varepsilon)$-dynamics of $S,Y$ given by
(\ref{eq:sedyn}), (\ref{eq:yedyn}), and with $\mathbb Q(0)$
corresponding to the MEMM $\mathbb Q^{0}$. We have now formulated the
indifference pricing control problem in the form of a control problem
akin to that described in Section \ref{sec:cp}. We then have the
following result.

\begin{theorem}
\label{thm:asexip}

Let the payoff of the claim, $F$, be a functional of the paths of
$S,Y$, satisfying Assumption \ref{ass:condF}. Let the $\mathbb
Q(\varepsilon)$-dynamics of $S,Y$ be given by
(\ref{eq:sedyn},\ref{eq:yedyn}), with $\mathbb Q(\varepsilon)$ given
by (\ref{eq:qeq0}), and with the parameter $\varepsilon$ given by
(\ref{eq:epsalpha}). Then for small risk aversion $\alpha$, the
indifference price process of the claim has the asymptotic expansion
\begin{equation}
p_{t}(\alpha) = \mathbb E^{\mathbb Q^{0}}[F|\mathcal{F}_{t}] +
\frac{1}{2}{\alpha}\mathbb E^{\mathbb
Q^{0}}\left[\left.\int_{t}^{T}\Vert\xi^{(0)}_{u}\Vert^{2}\ud
u\right\vert\mathcal{F}_{t}\right] + O(\alpha^{2}), \quad 0\leq t\leq T,
\label{eq:asexip}
\end{equation}
where $\mathbb Q^{0}$ is the minimal entropy martingale measure, and
$\xi^{(0)}$ is the process in the Kunita-Watanabe decomposition
(\ref{eq:gkwde}) of the claim, under $\mathbb Q(0)\equiv\mathbb
Q^{0}$. 

\end{theorem}

\begin{proof}

In the state dynamics (\ref{eq:sedyn},\ref{eq:yedyn}) each choice of
the perturbation $\varepsilon\varphi$ gives rise to a different
measure $\mathbb Q(\varepsilon)$. To apply Theorem \ref{thm:asex}, we
fix a measure $\mathbb M$ and instead consider the perturbed state
process
$X^{(\varepsilon)}=(S^{(\varepsilon)},Y^{(\varepsilon)})^{\top}$, with
dynamics under $\mathbb M$ given by 
\begin{eqnarray*}
\ud S^{(\varepsilon)}_{t} & = &
\diag_{d}(S^{(\varepsilon)}_{t})\sigma_{t}\ud W^{\mathbb
M}_{t}, \\  
\ud Y^{(\varepsilon)}_{t} & = & \diag_{m-d}(Y^{(\varepsilon)}_{t})[(\mu^{Y}_{t} -
\beta_{t}q^{0}_{t})\ud t + \beta_{t}(\ud W^{\mathbb M}_{t} +
\varepsilon\varphi_{t}\ud t)],
\end{eqnarray*}
for some $m$-dimensional $\mathbb M$-Brownian motion $W^{\mathbb M}$.
The dynamics of the state variable $X^{(\varepsilon)}$ under $\mathbb
M$ match those of $(S,Y)^{\top}$ under $\mathbb Q(\varepsilon)$, and
are of the required form (\ref{eq:xeps2}), with $\varepsilon=0$
corresponding to the MEMM $\mathbb Q^{0}$. We can now apply Theorem
\ref{thm:asex} directly, with the Kunita-Watanabe decomposition
(\ref{eq:gkwde}) of the claim under $\mathbb Q(0)\equiv\mathbb Q^{0}$
taking the place of the martingale representation result
(\ref{eq:mrF}), and the result duly follows.

\end{proof}

The underlying message of Theorem \ref{thm:asexip} is that for small
risk aversion, the lowest order contribution to the indifference price
process is the marginal utility-based price process
$\widehat{p}_{t}:=\mathbb E^{\mathbb Q^{0}}[F|\mathcal{F}_{t}]$,
corresponding to the valuation methodology developed by Davis
\cite{mhad97}. The first order correction is a mean-variance
correction, since the Kunita-Watanabe decomposition (\ref{eq:gkwde})
for $\varepsilon=0$ is the F\"ollmer-Schweizer-Sondermann decomposition of
the claim under $\mathbb Q^{0}$, and the integrand $\theta^{(0)}$ in
(\ref{eq:gkwde}) is a risk-minimising strategy in the sense of
F\"ollmer and Sondermann \cite{fs86} under $\mathbb Q^{0}$. Similar
results have been obtained for bounded claims by Mania and Schweizer
\cite{manschw05} and Kallsen and Rheinl\"ander \cite{kr11}. The
contribution here is to show a new methodology for obtaining this
result, for a square-integrable claim.  The strategy $\theta^{(0)}$ is,
in general, the zero risk aversion limit of the optimal hedging
strategy $\theta(\alpha)$ (see, for example, \cite{manschw05,kr11} for
a bounded claim), and hence can also be interpreted as the marginal
utility-based hedging strategy.

Note that using (\ref{eq:gkwde}) for $\varepsilon=0$, we can write 
(\ref{eq:asexip}) as
\begin{equation}
p_{t}(\alpha) = \mathbb E^{\mathbb Q^{0}}[F|\mathcal{F}_{t}] 
+ \frac{1}{2}{\alpha}\left(\var^{\mathbb Q^{0}}[F|\mathcal{F}_{t}] -
\mathbb E^{\mathbb Q^{0}}\left[\left.\int_{t}^{T}\Vert\theta^{(0)}_{u}\Vert^{2}\ud
[S]_{u}\right\vert\mathcal{F}_{t}\right]\right) 
+ O(\alpha^{2}),
\label{eq:aealt}  
\end{equation}
for $t\in[0,T]$, which highlights the mean-variance structure of the
asymptotic representation.

\section{Applications}
\label{sec:examples}

Here we show some examples where Theorem \ref{thm:asexip} would
apply. In these examples we assume that the functional $F$ satisfies
Assumption \ref{ass:condF}. This is a relatively mild assumption and
would apply in a wide range of models, but of course would need to be
checked on a case-by-case basis in specific models, and would depend
on the model and also on the specific form of the functional $F$. We
give a concrete case in Example \ref{examp:lbrm} of a lookback put
option on a non-traded asset, where we check that Assumption
\ref{ass:condF} is satisfied.

\begin{example}[Multi-dimensional random parameter basis risk model]
\label{examp:mdbrm}

This is the model of Section \ref{sec:iviim}, with $d$ traded stocks
$S$ and $(m-d)$ non-traded assets $Y$, and with the volatility process
$\sigma$ in (\ref{eq:sdyn}) given by
\begin{equation*}
\sigma_{t} = \left(\begin{array}{cc}\sigma^{S}_{t} &
\mathbf{0}_{d\times(m-d)}\end{array}\right), \quad 0\leq t\leq T,
\end{equation*}
where $\sigma^{S}$ is a $d\times d$ invertible matrix process, and
where $\mathbf{0}_{d\times(m-d)}$ denotes the zero $d\times(m-d)$
matrix. Write the $m$-dimensional Brownian motion $W$ as
$W=(W^{S},W^{S,\perp})^{\top}$, where $W^{S}$ denotes the first $d$
components of $W$. Then the $d$ traded stocks are driven by $d$
Brownian motions, and the non-traded assets are imperfectly correlated
with $S$. The claim payoff $F$ is typically dependent on the evolution
of $Y$ only, though our results are valid without this restriction.

In this case, the process $\lambda$ in (\ref{eq:lambda}) and the
integrand $q$ in (\ref{eq:ZmbbQ}) are given by
\begin{equation*}
\lambda_{t} = \binom{\lambda^{S}_{t}}{\mathbf{0}_{m-d}}, \quad q_{t} =
\binom{\lambda^{S}_{t}}{\gamma_{t}}, \quad 0\leq t\leq T,  
\end{equation*}
where $\lambda^{S}$ is the stocks' $d$-dimensional market price of
risk process, given by $\lambda^{S}:=(\sigma^{S})^{-1}\mu^{S}$, and
$\gamma$ is an $(m-d)$-dimensional adapted process. Each choice of
$\gamma$ leads to a different ELMM $\mathbb Q$, with
$\gamma=\mathbf{0}_{m-d}$ corresponding to the minimal martingale
measure $\mathbb Q_{M}$, and $\gamma=\gamma^{0}$ corresponding to the
minimal entropy martingale measure $\mathbb Q^{0}\equiv\mathbb Q_{E}$,
for some $(m-d)$-dimensional process $\gamma^{0}$. The density process
of any ELMM $\mathbb Q\in\mathbf{M}_{f}$ is then given by
\begin{equation}
Z^{\mathbb Q}_{t} = \mathcal{E}(-\lambda^{S}\cdot W^{S} - \gamma\cdot
W^{S,\perp})_{t}, \quad 0\leq t\leq T.  
\label{eq:ZQmdbrm}
\end{equation}
The indifference price expansion of the claim with payoff $F$ is then
of the form (\ref{eq:asexip}) or, equivalently, (\ref{eq:aealt}).

A special feature of these models arises when the process
$\lambda^{S}$ is either deterministic or does not depend on the
non-traded asset prices $Y$. In this case it is not hard to see that
the MEMM $\mathbb Q^{0}=\mathbb Q_{M}$. This is because the relative
entropy process between $\mathbb Q\in\mathbf{M}_{f}$ and $\mathbb P$
is given by
\begin{equation}
I_{t}(\mathbb Q|\mathbb P) = \mathbb E^{\mathbb
Q}\left[\left.\frac{1}{2}\int_{t}^{T}(\Vert\lambda^{S}_{u}\Vert^{2} +
\Vert\gamma_{u}\Vert^{2})\ud u\right\vert\mathcal{F}_{t}\right], \quad 0\leq
t\leq T.
\label{eq:basisent}
\end{equation}
The problem of finding the minimal entropy martingale measure is then
to minimise this functional subject to $\mathbb Q$-dynamics of $S,Y$
given by (\ref{eq:sqdyn},\ref{eq:yqdyn}), with the process $\gamma$
playing the role of a control. In the current notation, the $\mathbb
Q$-dynamics of $Y$ are
\begin{equation*}
\ud Y_{t} = \diag_{m-d}(Y_{t})\left[\left(\mu^{Y}_{t} -
\beta_{t}\binom{\lambda^{S}_{t}}{\gamma_{t}}\right)\ud t +
\beta_{t}\ud W^{\mathbb Q}_{t}\right].
\end{equation*}
From this it is clear that if $\lambda^{S}$ does not depend on $Y$,
then it is unaffected by the control, and then the relative entropy
process in (\ref{eq:basisent}) is minimised by choosing
$\gamma=\mathbf{0}_{m-d}$, so $\mathbb Q_{E}\equiv\mathbb Q^{0}=\mathbb
Q_{M}$. In this case, the Kunita-Watanabe decomposition of the claim
under $\mathbb Q_{M}$ will be of the form
\begin{equation*}
F = \mathbb E^{\mathbb Q_{M}}[F] + (\theta^{M}\cdot S)_{T} +
(\xi^{M}\cdot W^{\mathbb Q_{M}})_{T}, 
\end{equation*}
for integrands $\theta^{M},\xi^{M}$ such that the $\mathbb
Q_{M}$-martingales $(\theta^{M}\cdot S)$ and $(\xi^{M}\cdot W^{\mathbb
  Q_{M}})$ are orthogonal, and $W^{\mathbb Q_{M}}$ is a $\mathbb
Q_{M}$-Brownian motion. An example where this pertains is given in
Monoyios \cite{mmamf10}, in a two-dimensional model of basis risk with
partial information. The indifference price process expansion is
given by the analogue of (\ref{eq:aealt}), as
\begin{equation}
p_{t}(\alpha) = \mathbb E^{\mathbb Q_{M}}[F|\mathcal{F}_{t}] +
\frac{1}{2}\alpha\left( 
\var^{\mathbb Q_{M}}[F|\mathcal{F}_{t}] - \mathbb E^{\mathbb 
Q_{M}}\left[\int_{t}^{T}\Vert\theta^{M}_{u}\Vert^{2}\ud 
[S]_{u}\right]\right) + O(\alpha^{2}), \quad 0\leq t\leq T.
\label{eq:t0ipex}
\end{equation}
When the model is Markovian, the integrand $\theta^{M}$ can sometimes
be expressed in terms of the partial derivatives with respect to $S$
and $Y$ of the marginal price process
$\widehat{p}(t,S_{t},Y_{t})=\mathbb E^{\mathbb
Q_{M}}[F|S_{t},Y_{t}]$. An example where this is carried out can be
found in Monoyios \cite{mmamf10}.

\end{example}

\begin{example}[Two-dimensional random parameter basis risk model]
\label{examp:lbrm}

This is a random parameter version of the classical example first
considered by Davis \cite{mhad06}, and so a two-dimensional case of
Example \ref{examp:mdbrm}. We show how the asymptotic expansion
for the indifference price simplifies in this case, and so we extend
results of Davis and others \cite{mhad06,hend02,mmqf04} to general
(so possibly path-dependent) payoffs dependent on the non-traded
asset price, in a random parameter scenario. We also illustrate how the
conditions in Assumption \ref{ass:condF} are satisfied in the case of
a lookback put option in the constant parameter (lognormal) case.

Set $d=1$, $m=2$ in Example \ref{examp:mdbrm}, and set 
\begin{equation*}
\beta_{t} = \sigma^{Y}_{t}\left(\begin{array}{cc} \rho_{t} &
\sqrt{1-\rho^{2}_{t}}\end{array}\right),  \quad \rho_{t}\in(-1,1),
\quad 0\leq t\leq T, 
\end{equation*}
for adapted processes $\sigma^{Y},\rho$. Then the stock and non-traded asset
are imperfectly correlated with cross-variation process given by
\begin{equation*}
[S,Y]_{t} = \int_{0}^{t}\rho_{u}\sigma^{S}_{u}\sigma^{Y}_{u}S_{u}Y_{u}\ud u, \quad
0\leq t\leq T.   
\end{equation*}
The $\mathbb P$-dynamics of the assets are 
\begin{equation*}
\ud S_{t} = \sigma^{S}_{t}S_{t}(\lambda^{S}_{t}\ud t + \ud W^{S}_{t}), \quad
\ud Y_{t} = Y_{t}(\mu^{Y}_{t}\ud t + \sigma^{Y}_{t}\ud W^{Y}_{t}), \quad
\lambda^{S}_{t}:=\mu^{S}_{t}/\sigma^{S}_{t}, \quad 0\leq t\leq T, 
\end{equation*}
where $W^{Y}=\rho W^{S}+\sqrt{1-\rho^{2}}W^{S,\perp}$. 

The density process of any ELMM $\mathbb Q\in\mathbf{M}_{f}$ is once
again given by (\ref{eq:ZQmdbrm}), with
$\lambda^{S}_{t}=\mu^{S}_{t}/\sigma^{S}_{t}$, and in this case the
processes in the Dol\'eans exponential are one-dimensional.  For the
MEMM $\mathbb Q^{0}\equiv \mathbb Q_{E}$, the integrand $\gamma$ in
(\ref{eq:ZQmdbrm}) is given by some process $\gamma^{0}$. We may write
the $\mathbb Q$-dynamics of $Y$ as a perturbation to the $\mathbb
Q^{0}$-dynamics, in the same manner as in Section
\ref{sec:iviim}. This gives the $\mathbb Q$-dynamics of the asset
prices in the form
\begin{equation}
\ud S_{t} = \sigma^{S}_{t}S_{t}\ud W^{S,\mathbb Q}_{t}, 
\quad \ud Y_{t} = Y_{t}\left[\nu_{t}\ud t + \sigma^{Y}_{t}\left(\ud W^{Y,\mathbb
Q}_{t} + \sqrt{1-\rho^{2}_{t}}\varepsilon\varphi_{t}\ud t\right)\right],  
\label{eq:Qdynassets}
\end{equation}
for $\mathbb Q$-Brownian motions $W^{S,\mathbb Q},W^{Y,\mathbb Q}$
with instantaneous correlation $\rho$, so that 
\begin{equation}
W^{Y,\mathbb Q} = \rho W^{S,\mathbb Q} + \sqrt{1-\rho^{2}}W^{S,\perp,\mathbb Q},  
\label{eq:wyq}
\end{equation}
with $W^{S,\mathbb Q},W^{S,\perp,\mathbb Q}$ independent $\mathbb
Q$-Brownian motions,
$\nu:=\mu^{Y}-\sigma^{Y}(\rho\lambda^{S}+\sqrt{1-\rho^{2}}\gamma^{0})$
and $\varepsilon\varphi:=-(\gamma-\gamma^{0})$, for a small parameter
$\varepsilon$ and control process $\varphi$. For
$\varepsilon\varphi=0$ we have the dynamics under the MEMM $\mathbb
Q^{0}$. Once again, the perturbation expansion for the indifference
price of a claim with payoff $F$ depending on the evolution of $S,Y$
over $[0,T]$ will be of the form (\ref{eq:asexip}) or, equivalently,
(\ref{eq:aealt}). In the case that $\lambda^{S}$ has no dependence on
$Y$, then $\gamma^{0}=0$ and $\mathbb Q^{0}=\mathbb Q_{M}$.

Another special case arises when $\rho$ is deterministic (say,
constant), $\lambda^{S},\mu^{Y},\sigma^{Y}$ are adapted to the
filtration generated by $W^{Y}$, so depend on the evolution of the
non-traded asset price only, and the claim is written on the
non-traded asset, so its payoff $F$ also depends only on the evolution
of $Y$. (This would also apply in a stochastic volatility model where
$Y$ is the process driving the volatility, and then $F$ would be a
volatility derivative.) In this case the Kunita-Watanabe decomposition
of $F$ under $\mathbb Q^{0}$ will be of the special form
\begin{equation}
F = \mathbb E^{\mathbb Q^{0}}[F] + (\psi^{(0)}\cdot W^{Y,\mathbb Q^{0}})_{T},  
\label{eq:kwd1}
\end{equation}
for some process $\psi^{(0)}$ such that $(\psi^{(0)}\cdot W^{Y,\mathbb
  Q^{0}})$ is a $\mathbb Q^{0}$-martingale. But we also have the
general form (\ref{eq:gkwde}) of this decomposition, which in this
case reads as
\begin{equation}
F = \mathbb E^{\mathbb Q^{0}}[F] + (\theta^{(0)}\cdot S)_{T} +
(\xi^{(0)}\cdot W^{S,\perp,\mathbb Q^{0}})_{T},
\label{eq:kwd2}
\end{equation}
for integrands $\theta^{(0)},\xi^{(0)}$ (here, $\theta^{(0)}$ would be
the marginal utility-based hedging strategy for the claim).

Equating the representations in (\ref{eq:kwd1}) and (\ref{eq:kwd2})
and in view of (\ref{eq:Qdynassets}) and (\ref{eq:wyq}) for the case
$\mathbb Q=\mathbb Q^{0}$, it is easy to see that
$\theta^{(0)},\xi^{(0)}$ are both linearly related to the process
$\psi^{(0)}$, through
\begin{equation*}
\theta^{(0)}\sigma^{S}S = \rho\psi^{(0)}, \quad \xi^{(0)} =
\sqrt{1-\rho^{2}}\psi^{(0)}.   
\end{equation*}
It is then straightforward to compute that
\begin{equation*}
\var^{\mathbb Q^{0}}[F] = \frac{1}{\rho^{2}}\mathbb E^{\mathbb
  Q^{0}}\left[\int_{0}^{T}(\theta^{(0)}_{t})^{2}\ud[S]_{t}\right].   
\end{equation*}
The time-zero indifference price expansion in this case then
simplifies to
\begin{equation*}
p_{0}(\alpha) = \mathbb E^{\mathbb Q^{0}}[F] +
\frac{1}{2}\alpha(1-\rho^{2})\var^{\mathbb Q^{0}}[F] + O(\alpha^{2}), 
\end{equation*}
which is an extension of the form found in \cite{mhad06,hend02,mmqf04}
to European payoffs $F$ satisfying Assumption \ref{ass:condF}, in
models with random parameters dependent on $Y$. If, in addition,
$\lambda^{S}$ is deterministic, then $\mathbb Q^{0}=\mathbb Q_{M}$.

It is instructive to see how Assumption \ref{ass:condF} would be
checked in a simple case of this example. Suppose the parameters of
the model are constants, so that $Y$ is a geometric Brownian motion.
Let the claim be a European floating strike lookback put option on the
non-traded asset, so that $F$ is a functional of a one-dimensional
Brownian motion given by
\begin{equation*}
F = \max_{0\leq t\leq T}Y_{t} - Y_{T},
\end{equation*}
To ease notation, write $W\equiv W^{Y,\mathbb Q}$ for the Brownian
motion driving $Y$ under any ELMM. When the perturbation
$\varepsilon\varphi$ is zero, $Y$ satisfies
\begin{equation*}
\ud Y_{t} = Y_{t}(\nu\ud t + \eta\ud W_{t}), \quad Y_{0} = y
\in\mathbb R^{+},   
\end{equation*}
for constants $\nu$ and $\eta>0$. For concreteness, let us suppose
that $\nu-\frac{1}{2}\eta^{2}>0$. The functional $F\equiv F(W)$ is
given by
\begin{equation*}
F(W) = y\exp\left[\left(\nu -
\frac{1}{2}\eta^{2}\right)T\right]\left(\exp\left(\eta\max_{0\leq 
t\leq T}W_{t}\right) - \exp(\eta W_{T})\right).    
\end{equation*}
Consider the two functionals $F_{1}(W):=\exp(\eta W_{T})$ and
$F_{2}(W):=\exp\left(\eta\max_{0\leq t\leq T}W_{t}\right)$ in turn. 

For $F_{1}$, it is straightforward to see square-integrability, and
that Assumption \ref{ass:condF} (ii) is satisfied with $k=F{1}$ and
$g(\varepsilon)=\exp(\eta\varepsilon)-1$. It is also easy to compute
\begin{equation*}
\lim_{\varepsilon\to
0}\frac{1}{\varepsilon}\left[F_{1}\left(W +
\varepsilon\int_{0}^{\cdot}\varphi_{s}\ud s\right) - F_{1}(W)\right] =
\eta Y_{T}\int_{0}^{T}\varphi_{t}\ud t = \int_{0}^{T}\eta
F_{1}(W)\varphi_{t}\ud t. 
\end{equation*}
Therefore, $\partial F_{1}(W;(t,T])=\eta F_{1}(W)$. 

For $F_{2}$, the maximum of the Brownian motion over $[0,T]$ is
achieved at some random time $\tau(\omega)\equiv\tau(W)$, so in this
case we have
\begin{equation*}
F_{2}(W) = y\exp\left[\left(\nu -
\frac{1}{2}\eta^{2}\right)T + \eta W_{\tau(W)}\right].   
\end{equation*}
The first two conditions in Assumption \ref{ass:condF} are satisfied
in a similar manner as for $F_{1}$. For the last condition, 
with $\Phi=\int_{0}^{\cdot}\varphi_{s}\ud s$  we obtain
\begin{equation*}
\lim_{\varepsilon\to
0}\frac{1}{\varepsilon}\left[F_{2}\left(W +
\varepsilon\Phi\right) - F_{2}(W)\right] =
\eta F_{2}(W)\Phi_{\tau(W)} =
\int_{0}^{T}\eta F_{2}(W)\mathbbm{1}_{\{\tau(W)>t\}}\varphi_{t}\ud t.  
\end{equation*}
Therefore, $\partial F_{2}(W;(t,T])=\eta
F_{2}(W)\mathbbm{1}_{\{\tau(W)>t\}}$. This shows how Assumption
\ref{ass:condF} is compatible with path-dependent payoffs. Similar
reasoning can work with random parameter models.

\end{example}

\begin{example}[Basis risk with stochastic correlation]

This model has been considered by Ankirchner and Heyne \cite{ah12},
who examined local risk minimisation methods for hedging basis risk. A
traded asset $S$ and non-traded asset $Y$ follow correlated geometric
Brownian motions, as in Example \ref{examp:lbrm}, but the correlation
$\rho=(\rho_{t})_{0\leq t\leq T}$ is now stochastic. In this case, we
have $m=3$, $d=1$. With $W$ a three-dimensional Brownian motion, let
$W^{S}=W^{1}$, $W^{Y}=\rho W^{1}+\sqrt{1-\rho^{2}}W^{2}$,
$W^{\rho}=\delta W^{1}+\eta W^{2}+\sqrt{1-\delta^{2}-\eta^{2}}W^{3}$,
for constants $\delta,\eta$ such that $\delta^{2}+\eta^{2}\leq 1$. The
state variable dynamics are then
\begin{eqnarray*}
\ud S_{t} & = & \sigma^{S}S_{t}(\lambda^{S}\ud t + \ud W^{1}_{t}), \\
\ud Y_{t} & = & Y_{t}[\mu^{Y}\ud t + \sigma^{Y}(\rho_{t}\ud W^{1}_{t} +
\sqrt{1-\rho^{2}_{t}}\ud W^{2}_{t})], \\
\ud\rho_{t} & = & g_{t}\ud t + h_{t}(\delta\ud W^{1}_{t} + \eta\ud
W^{2}_{t} + \sqrt{1-\delta^{2}-\eta^{2}}\ud W^{3}_{t}).     
\end{eqnarray*}
Here, $g,h$ are adapted processes such that $\rho_{t}\in[-1,1]$ almost
surely. Ankirchner and Heyne \cite{ah12} give some specific examples
of such models.

In this example we also have $\mathbb Q^{0}=\mathbb Q_{M}$, with
$Z^{\mathbb Q_{M}}=\mathcal{E}(-\lambda^{S}W^{1})$, and the
F\"ollmer-Schweizer-Sondermann decomposition of the claim is of the form
\begin{equation}
F = \mathbb E^{\mathbb Q_{M}}[F] + (\theta^{M}\cdot S)_{T} +
(\xi^{M}\cdot W^{2,\mathbb Q_{M}})_{T} + (\phi^{M}\cdot W^{3,\mathbb Q_{M}})_{T},    
\label{eq:fsssc}
\end{equation}
for some integrands $\theta^{M},\xi^{M},\phi^{M}$, where $W^{\mathbb
  Q_{M}}=(W^{1,\mathbb Q_{M}},W^{2,\mathbb Q_{M}},W^{3,\mathbb
  Q_{M}})^{\top}$ is a three-dimensional $\mathbb Q_{M}$-Brownian
motion, the first of which drives the stock, so that the stochastic
integrals in (\ref{eq:fsssc}) are orthogonal $\mathbb
Q_{M}$-martingales. The time-zero indifference price expansion is
again of the form (\ref{eq:t0ipex}).
  
\end{example}

Many examples are covered by the framework of Theorem
\ref{thm:asexip}, including classical stochastic volatility models,
basis risk models with stochastic volatility (so $m=3$, $d=1$) with a
traded and non-traded asset both driven by a common stochastic
volatility process (and stochastic correlation can be added to this
framework), or basis risk models with unknown asset drifts, extending
\cite{mmamf10} (which modelled the drifts as unknown constants) to
model the drifts as linear diffusions.

\subsection{Entropy minimisation in stochastic volatility models}
\label{subsec:emsvm}

We end with another application of the asymptotic methods developed in
the paper. This time, we are interested in finding the minimal entropy
martingale measure $\mathbb Q^{0}\equiv\mathbb Q_{E}$ in a stochastic
volatility model. A traded asset $S$ and a non-traded stochastic
factor $Y$ follow, under the physical measure $\mathbb P$,
\begin{eqnarray}
\ud S_{t} & = & \sigma(Y_{t})S_{t}\left(\lambda(Y_{t})\ud t + \ud W^{S}_{t}\right),
\label{eq:S} \\ 
\ud Y_{t} & = & a(Y_{t})\ud t + b(Y_{t})\ud W^{Y}_{t}, \label{eq:Y} 
\end{eqnarray}
for suitable functions $\sigma,\lambda,a,b$ such that there are unique
strong solutions to (\ref{eq:S},\ref{eq:Y}). The Brownian motions
$W^{S},W^{Y}$ have constant correlation $\rho\in[-1,1]$. We write
$W^{Y}_{t}=\rho W^{S}_{t} + \sqrt{1-\rho^{2}}W^{S,\perp}_{t}$. The
density process of any ELMM $\mathbb Q$ is
\begin{equation*}
Z^{\mathbb Q}_{t} = \mathcal{E}(-\lambda\cdot W^{S} - \gamma\cdot
W^{S,\perp})_{t}, \quad 0\leq t\leq T,  
\end{equation*}
for some square-integrable process $\gamma$ such that $Z^{\mathbb Q}$
is a $\mathbb P$-martingale.

The entropy minimisation problem is the stochastic control problem to
minimise
\begin{equation*}
I_{0}(\mathbb Q|\mathbb P) = \mathbb E^{\mathbb
  Q}\left[\frac{1}{2}\int_{0}^{T}(\lambda^{2}(Y_{t}) +
  \gamma^{2}_{t})\ud t\right],   
\end{equation*}
over control processes $\gamma$, where we assume that $I_{0}(\mathbb
Q|\mathbb P)<\infty$, and where, under $\mathbb Q$, $S,Y$ follow
\begin{eqnarray}
\ud S_{t} & = & \sigma(Y_{t})S_{t}\ud W^{S,\mathbb Q}_{t}, \nonumber \\
\ud Y_{t} & = & (a(Y_{t}) - b(Y_{t})\rho\lambda(Y_{t}))\ud t + b(Y_{t})(\ud
W^{Y,\mathbb Q}_{t} - \sqrt{1-\rho^{2}}\gamma_{t}\ud t), 
\label{eq:svyqdyn}
\end{eqnarray}
for $\mathbb Q$-Brownian motions $W^{S,\mathbb Q},W^{Y,\mathbb Q}$
with correlation $\rho$, such that setting $\gamma=0$ yields the
minimal martingale measure $\mathbb Q_{M}$.

The idea here is to consider the drift adjustment
$\sqrt{1-\rho^{2}}\gamma_{t}$ in (\ref{eq:svyqdyn}) as a perturbation
to the Brownian paths, and hence to convert the entropy minimisation
problem to the type of control problem we have considered in Section
\ref{sec:cp}, in the limit that the absolute value of the correlation
is close to $1$, so $1-\rho^{2}$ is small. To this end, we define a
parameter $\varepsilon$ and a control process $\varphi$ such that
\begin{equation*}
\varepsilon^{2} = 1-\rho^{2}, \quad \varepsilon\varphi =
-\sqrt{1-\rho^{2}}\gamma, 
\end{equation*}
and we define a parametrised family of measures $\{\mathbb
Q(\varepsilon)\}_{\varepsilon\in\mathbb R}$, such that
\begin{equation*}
\mathbb Q = \mathbb Q(\varepsilon), \quad \mathbb Q(0) = \mathbb Q_{M}.   
\end{equation*}
The state variable dynamics for $Y$ are then given by
\begin{equation}
\ud Y_{t} = (a(Y_{t}) - b(Y_{t})\rho\lambda(Y_{t}))\ud t + b(Y_{t})(\ud
W^{Y,\mathbb Q(\varepsilon)}_{t} + \varepsilon\varphi_{t}\ud t).   
\label{eq:yphi}
\end{equation}
With $\Phi:=\int_{0}^{\cdot}\varphi_{s}\ud s$, we define a
square-integrable functional $F\equiv F(W^{\mathbb
Q(\varepsilon)}+\varepsilon\Phi)$ of the Brownian paths by
\begin{equation*}
F := \frac{1}{2}\int_{0}^{T}\lambda^{2}(Y_{t})\ud t =: \frac{1}{2}K_{T},  
\end{equation*}
where, for brevity of notation, we have defined the so-called
mean-variance trade-off process $K$ by
\begin{equation}
K_{t} := \int_{0}^{t}\lambda^{2}(Y_{u})\ud u, \quad 0\leq t\leq T.  
\label{eq:Kt}
\end{equation}
We assume that the model is such that $K_{T}$ defines a functional
satisfying Assumption \ref{ass:condF}.

In this notation, the relative entropy between the minimal martingale
measure and $\mathbb P$ is
\begin{equation}
I_{0}(\mathbb Q_{M}|\mathbb P) = \mathbb E^{\mathbb
  Q_{M}}\left[\frac{1}{2}K_{T}\right] = \mathbb E^{\mathbb
  Q(0)}[F(W^{\mathbb Q(0)})].
\label{eq:fie}
\end{equation}
The control
problem to minimise $I_{0}(\mathbb Q|\mathbb P)$ over ELMMs $\mathbb
Q\in\mathbf{M}_{f}$ then has value function
\begin{equation}
I_{0}(\mathbb Q_{E}|\mathbb P) :=
\inf_{\varphi\in\mathcal{A}(\mathbf M_{f})}\mathbb E^{\mathbb
Q(\varepsilon)}\left[F\left(W^{\mathbb
Q(\varepsilon)} + \varepsilon\int_{0}^{\cdot}\varphi_{s}\ud s\right) + 
\frac{1}{2}\int_{0}^{T}\varphi^{2}_{t}\ud t\right], 
\label{eq:objfi}
\end{equation}
where $\mathcal{A}(\mathbf M_{f})$ denotes the set of controls
$\varphi$ such that $I_{0}(\mathbb Q|\mathbb P)$ is finite.

We have now formulated the entropy minimisation problem in the form we
need to be able to apply the Malliavin asymptotic method, and this
gives the theorem below.

\begin{theorem}
\label{thm:sventm}

In the stochastic volatility model defined by (\ref{eq:S},\ref{eq:Y}),
suppose the terminal value $K_{T}$ of mean-variance trade-off process
in (\ref{eq:Kt}) defines a Brownian functional satisfying Assummption
\ref{ass:condF}. Then the relative entropy between the minimal
entropy martingale measure $\mathbb Q_{E}$ and $\mathbb P$, in the
limit that $1-\rho^{2}$ is close to $1$, is given as
\begin{equation*}
I_{0}(\mathbb Q_{E}|\mathbb P) = I_{0}(\mathbb Q_{M}|\mathbb P) -
\frac{1}{8}(1-\rho^{2})\var^{\mathbb Q_{M}}[K_{T}] + O((1-\rho^{2})^{2}),
\end{equation*}
where $\mathbb Q_{M}$ is the minimal martingale measure.

\end{theorem}

\begin{proof}

This is along the same lines as previous results, so we only sketch
the details. One appeals to the decomposition of $F$ under
$\mathbb Q(0)$, which is of the form
\begin{equation}
F(W^{\mathbb Q(0)}) =  \mathbb E^{\mathbb Q(0)}[F(W^{\mathbb Q(0)})] + 
(\xi^{(0)}\cdot W^{\mathbb Q(0)})_{T}, 
\label{eq:kwdsv}
\end{equation}
for some integrand $\xi^{(0)}$. Such a decomposition exists uniquely,
given that $F$ depends only on $Y$, and the dynamics in
(\ref{eq:yphi}). We expand the objective function (\ref{eq:objfi})
about $\varepsilon=0$ and use the representation
(\ref{eq:kwdsv}). This gives
\begin{eqnarray*}
&& \mathbb E^{\mathbb Q(\varepsilon)}\left[F\left(W^{\mathbb Q(\varepsilon)} +
  \varepsilon\int_{0}^{\cdot}\varphi_{s}\ud s\right) +
\frac{1}{2}\int_{0}^{T}\varphi^{2}_{t}\ud t\right] \\
& = & \mathbb
E^{\mathbb Q(0)}\left[F(W^{\mathbb Q(0)}) +
\int_{0}^{T}\left(\varepsilon\xi^{(0)}_{t}\varphi_{t} +
\frac{1}{2}\varphi^{2}_{t}\right)\ud t\right] + o(\varepsilon). 
\end{eqnarray*}
We minimise the right-hand-side over $\varphi$ by choosing
$\varphi=-\varepsilon\xi^{(0)}$. Using (\ref{eq:kwdsv}) again, and
recalling (\ref{eq:fie}), the result follows.
 
\end{proof}

\begin{remark}
\label{rem:esscher}

In \cite{mmspl07,mmdef06}, Esscher transform relations between
$\mathbb Q_{E}$ and $\mathbb Q_{M}$ are derived, and it is an exercise
in asymptotic analysis to see that those results are consistent with
Theorem \ref{thm:sventm}. 

\end{remark}

\section{Conclusions}

It is quite natural to apply Malliavin calculus ideas in stochastic
control problems where the control turns out to be a drift which is
considered as a perturbation to a Brownian motion, and this is the
path taken in this paper. We have shown how the method can yield
small risk aversion asymptotic expansions for exponential indifference
prices in It\^o process models, and how one can identify the minimal
entropy measure as a perturbation to the minimal martingale measure in
stochastic volatility models. It would be interesting to extend the
method to models with jumps in the underlying state process.

{\small
\bibliography{mallrefs_final}
\bibliographystyle{siam}
}

\end{document}